\documentclass[12pt]{article}
\usepackage{amsfonts}
\usepackage{amsmath,amssymb}
\usepackage{amsthm}
\usepackage{graphicx}
\usepackage{natbib}
\usepackage{url} 

\usepackage{xr}
\usepackage{caption}
\theoremstyle{plain}

\theoremstyle{plain}
\newtheorem{lemma}{Lemma}[section]
\theoremstyle{plain}
\newtheorem{assumption}{Assumption}[section]
\theoremstyle{plain}

\theoremstyle{plain}
\newtheorem{corollary}{Corollary}[section]
\theoremstyle{remark}

\theoremstyle{comment}

\theoremstyle{definition}

\theoremstyle{definition}
\newtheorem{example}{Example}[section]
\newcommand{\plim}{\operatornamewithlimits{plim\,}}

\newcommand{\E}{\mathbb{E}}

\mathchardef\mhyphen="2D
\def\y{y_1^n}
\def\x{x_1^n}

\def \t{\theta}

\def \E{\mathbb{E}}

\def \dt {\mathrm{d}}

\newcommand{\blind}{0}

\begin{document}

\def\spacingset#1{\renewcommand{\baselinestretch}%
{#1}\small\normalsize} \spacingset{1}


\if0\blind
{
  \title{\bf Variational Bayes in State Space Models: Inferential and Predictive
  	Accuracy}
  \author{David T. Frazier, Rub\'en Loaiza-Maya and Gael M. Martin\thanks{
    The authors gratefully acknowledge funding from the Australian Research Council.}\hspace{.2cm}\\
    Department of Econometrics and Business Statistics, Monash University}
  \maketitle
} \fi

\if1\blind
{
  \bigskip
  \bigskip
  \bigskip
  \begin{center}
    {\LARGE\bf Title}
\end{center}
  \medskip
} \fi

\bigskip
\begin{abstract}
Using theoretical and numerical\ results, we document the accuracy of
commonly applied variational Bayes methods across a range of state space
models. The results demonstrate that, in terms of accuracy on fixed
parameters, there is a clear hierarchy in terms of the methods, with
approaches that do not approximate the states yielding superior accuracy
over methods that do. {We} also document numerically{\ }that{\ }the
inferential discrepancies between the various methods often yield only small
discrepancies in predictive accuracy over small out-of-sample evaluation
periods. Nevertheless, in certain settings, these predictive discrepancies
can become meaningful over a longer out-of-sample period. This finding indicates
that the invariance of predictive results to inferential inaccuracy, which
has been an oft-touted point made by practitioners seeking to justify the
use of variational inference, is not ubiquitous and must be assessed on a
case-by-case basis.
\end{abstract}

\noindent%

\spacingset{1.5} 
\section{Introduction}


A common class of models used for time series modelling and prediction is
the class of {state space models (SSMs). This class includes nonlinear
	structures, like stochastic volatility models, regime switching models,
	mixture models, and models with random dynamic jumps; plus linear
	structures, such as linear Gaussian unobserved component models. (See {%
		\citealp{durbin2001}, \citealp{harvey2004},} and {\citealp{giordani2011},}
	for extensive reviews).}

The key feature of SSMs is their dependence on hidden, {or latent,} `local'\
variables, or states, which govern the dependence of the observed{\ data, in
	conjunction} with a vector of unknown `global'\ parameters. {This feature
	leads to inferential challenges with, for example, the likelihood function {%
		for the global parameters }being }analytically unavailable, except in {special cases. Whilst frequentist methods have certainly been adopted (see {\citealp{Daniel1993}, \citealp{RUIZ1994}, \citealp{anderson1996}, %
		\citealp{gallant1996}, \citealp{SANDMANN1998}, \citealp{Bates2006}, %
		\citealp{AITSAHALIA2007},} and {\citealp{AITSAHALIA2020}}, amongst others),
	it is arguable that Bayesian Markov chain Monte Carlo (MCMC) methods have
	become the most {common tool} for analysing general SSMs, with such
	techniques expanded in more recent times to accommodate particle filtering,
	via pseudo-marginal variants such as particle MCMC (PMCMC) ({%
		\citealp{andrieu:doucet:holenstein:2010}}; {\citealp{Flury2011}). }See \cite%
	{giordani2011} and \cite{fearnhead2011} for a detailed coverage of this
	literature, including the variety of MCMC-based algorithms adopted therein.}

Whilst (P)MCMC methods have been transformative in the SSM field, {they do
	suffer from certain well-known limitations. }Most notably, they require that
either the {(complete) }likelihood function is\ available in closed form or
that an unbiased estimator of it is available. Such methods also do not
necessarily scale well to high-dimensional problems;{\ that is, to models
	with multiple observed and/or state processes. }If the{\ assumed data
	generating process (DGP)} is intractable, inference can proceed using {%
	approximate Bayesian computation (ABC) (%
	\citealp{dean:singh:jasra:peters:2011}; \citealp{CREEL2015}; %
	\citealp{Martin2019}), since ABC requires only simulation - not evaluation }%
- of the DGP. However, ABC also does not scale well to problems with a large
number of parameters (see, e.g., Corollary 1 in %
\citealp{frazier2018asymptotic} for details).

Variational Bayes (VB) methods (see {\citealp{blei2017variational}} for a
review) can be seen as a potential class of alternatives to either (P)MCMC-
or ABC-based inference in SSMs. In particular, and in contrast to these
methods, VB scales well to high-dimensional problems, using
optimization-based techniques to effectively manage a large number of
unknowns (\citealp{tran2017variational}; \citealp{quiroz2018gaussian}; %
\citealp{koop2018variational}; \citealp{chan2020fast}; %
\citealp{loaiza2020fast}).

{In this paper}, we make three contributions to the literature on {the
	application of VB to SSMs}. The first contribution is to {highlight the
	fundamental issue that lies at the heart of the use of VB in an SSM setting,
	linking this to an existing issue identified }in the literature as the
`incidental parameter problem' (\citealp{Neyman1948}; %
\citealp{lancaster2000incidental}; \citealp{westling2019beyond}).{\ In
	brief, without due care, the application of VB to the local parameters in an
	SSM }leads to a lack of Bayesian consistency for the global parameters. {%
	Moreover}{, in a class of common SSMs, we {demonstrate analytically} the
	impact of this inconsistency on the resulting state inference, and show that
	even in idealized settings} inconsistent inference for {the global
	parameters can }lead to{\ highly inaccurate inferences about the local
	parameters. }The second contribution is to review some existing variational
methods, and to link their prospects for consistency to the manner in which
they do, or do not, circumvent th{e incidental parameter problem. Thirdly, }%
we {undertake a numerical} comparison of several competing variational
methods, in terms of {both }{inferential and predictive accuracy}. The key
findings are that: i) correct management of the local variables leads to
inferential accuracy that closely {matches that of exact (MCMC-based) Bayes;
	ii) }inadequate treatment of the local variables leads, in contrast, to
noticeably less accurate inference; iii) predictive accuracy {shows some
	robustness to} inferential inaccuracy, {but only for small sample sizes.
	Once the size of the sample is very large, the consistency (or otherwise) of
	a VB }method impinges on{\ predictive accuracy, }with a clear ranking
becoming evident across the methods for some {DGPs}; {with certain VB
	methods unable to produce similar out-of-sample accuracy results to exact
	Bayes {in some settings}. }

{We believe {that all three contributions serve as} novel }insights{\textbf{%
		\ }into the role of VB in SSMs, {which may} lead to best practice, if
	heeded. }

Throughout the remainder, we make use of the following notational
conventions. Generic $p,g$ are used to denote densities, and $\pi $ is used
to denote posteriors conditioned only on data, and where the conditioning is
made explicit depending on the situation. For any arbitrary collection of
data $(z_{1},\dots ,z_{n})$, we abbreviate this collection as $z_{1}^{n}$.
For a sequence $a_{n}$, the terms $O_{p}(a_{n})$ , $o_{p}(a_{n})$ and $%
\rightarrow _{p}$ have their usual meaning. Similarly, we let $%
\operatornamewithlimits{plim\,}_{n}X_{n}=c$ denote $X_{n}\rightarrow _{p}c$.
We let $d(\cdot ,\cdot )$ denote a metric on $\Theta \subseteq \mathbb{R}%
^{d_{\theta }}$. {The proofs of all theoretical results, certain
	definitions, {plus additional tables and figures}, are included in the
	Supplementary Appendix.}

\section{State space models: exact inference}

An SSM is a stochastic process consisting of the pair $\{(X_{t},Y_{t})\}$,
where $\{X_{t}\}$ is a Markov chain taking values in the measurable space $(%
\mathcal{X},\mathcal{F}_{X},\mu )$, and $\{Y_{t}\}$ is a process taking
values in a measure space $(\mathcal{Y},\mathcal{F}_{Y},\chi )$, such that,
conditional on $\{X_{t}\}$, the sequence $\{Y_{t}\}$ is independent. The
model is formulated through the following conditional and transition
densities: for a vector of unknown random parameters $\theta $ taking values
in the probability space $(\Theta ,\mathcal{F}_{\theta },P_{\theta })$,
where $P_{\theta }$ admits the density function $p_{\theta }$,%
\begin{eqnarray}
Y_{t}|X_{t},\theta &\sim &g_{\theta }(y_{t}|x_{t})  \label{eq:HMM_2} \\
X_{t+1}|X_{t},\theta &\sim &\chi _{\theta }(x_{t+1},x_{t}),  \label{eq:HMM_1}
\end{eqnarray}%
where $\chi _{\theta }(\cdot ,\cdot )$ denotes the transition kernel with
respect to the measure $\mu $. For simplicity, throughout the remainder we
disregard {the }terms dependence on the initial measure $\nu $ and the
invariant measure $\mu $, when no confusion will result. The order-one
Markov assumption for $X_{t}$ is innocuous, and any finite (and known)
Markov order can be accommodated via a redefinition of the state variables.

Given the independence of $Y_{t}$ conditional on $X_{t}$, and the Markovian
nature of $X_{t}|X_{t-1}$, the complete data likelihood is 
\begin{equation*}
p_{\theta }(y_{1}^{n},x_{1}^{n})=\nu (x_{1})g_{\theta
}(y_{1}|x_{1})\prod_{t=2}^{n}{{\chi _{\theta }(x_{t},x_{t-1})}}g_{\theta
}(y_{t}|x_{t}).
\end{equation*}%
The (average) observed data log-likelihood is thus 
\begin{equation}
\ell _{n}(\theta ):=\frac{1}{n}\log p_{\theta }(y_{1}^{n})=\frac{1}{n}\log
\int p_{\theta }(y_{1}^{n},x_{1}^{n})dx_{1}^{n},  \label{log_like}
\end{equation}%
and the maximum likelihood estimator (MLE) of $\theta $ is $\widehat{\theta }%
_{n}^{MLE}:=\operatornamewithlimits{argmax\,}_{\theta \in \Theta }\ell
_{n}(\theta )$. {As is standard knowledge,} $\ell _{n}(\theta )$ {is
	available in closed form only for particular forms of }$g_{\theta
}(y_{t}|x_{t})$ and $\chi _{\theta }(x_{t+1},x_{t});$ the canonical example
being when (\ref{eq:HMM_1}) and (\ref{eq:HMM_2}) define a linear Gaussian
state space model (LGSSM). Similarly, for $p(\theta )$ denoting the prior
density, the exact (marginal) posterior for $\theta $, defined as 
\begin{equation}
\pi (\theta |y_{1}^{n})=\int \pi (\theta ,x_{1}^{n}|y_{1}^{n})\dt x_{1}^{n},%
\text{ where }\pi (\theta ,x_{1}^{n}|y_{1}^{n})\propto
p(y_{1}^{n}|x_{1}^{n},\theta )p(x_{1}^{n}|\theta )p(\theta ),
\label{eq:posts}
\end{equation}%
{{is available (e.g. via straightforward MCMC methods) only in limited
		cases, the LGSSM being one such case. In more complex settings and/or
		settings where either }$\theta $ {or }$\{(X_{t},Y_{t})\}$, {or both, are
		high-dimensional, accessing (\ref{eq:posts}) can be difficult, with standard
		MCMC methods leading to slow mixing, and thus potentially unreliable
		inferences (\citealp{betancourt:2018}).}}

To circumvent these issues, recent research has suggested the use of
variational methods for SSMs: these methods can be used to approximate
either the log-likelihood function in (\ref{log_like}) or the marginal
posterior in (\ref{eq:posts}), depending on the mode of inference being
adopted. The focus of this paper, as {already highlighted}, is on
variational \textit{Bayes} and, in particular, on the accuracy of such
methods in SSMs. However, as part of the following section we also
demonstrate the asymptotic behaviour of frequentist variational {point
	estimators of }$\theta $, as this result will ultimately help us interpret
the behavior of the variational posterior in SSMs.

\section{State space models: variational inference\label{vi}}

\subsection{Overview}

{The idea of VB is }to produce an approximation to the joint posterior $\pi
(x_{1}^{n},\theta |y_{1}^{n})$ in \eqref{eq:posts} by searching over a given
family of distributions for the member that minimizes a user-chosen
divergence measure between the posterior of interest and the family. This {%
	replaces} the posterior sampling problem {with} one of optimization over the
family of densities used to implement the approximation. We now review the
use of variational methods in SSMs, paying particular attention to the
Markovian nature of the states.

{VB} approximates the posterior $\pi (x_{1}^{n},\theta |y_{1}^{n})$ by
minimizing the KL divergence between a family of densities {$\mathcal{Q}$},
with generic element $q(x_{1}^{n},\theta )$, and $\pi $: 
\begin{equation}
\text{KL}(q||\pi )=\int q(x_{1}^{n},\theta )\log \frac{q(x_{1}^{n},\theta )}{%
	\pi (x_{1}^{n},\theta |y_{1}^{n})}\dt x_{1}^{n}\dt\theta .  \label{kl}
\end{equation}%
Optimizing the KL divergence directly is not feasible {since it depends} on
the unknown $\pi (x_{1}^{n},\theta |y_{1}^{n})$; the very quantity we are
trying to approximate. However, minimizing the KL divergence between $q$ and 
$\pi $ is equivalent to maximizing the {so-called }variational evidence
lower bound (ELBO):%
\begin{equation}
\text{ELBO}(q||\pi ):=\int q(x_{1}^{n},\theta )\log \frac{%
	p(y_{1}^{n}|x_{1}^{n},\theta )p(x_{1}^{n}|\theta )p(\theta )}{%
	q(x_{1}^{n},\theta )}\dt x_{1}^{n}\dt\theta \text{,}  \label{elbo}
\end{equation}%
which we can access. {Hence}, for a given class $\mathcal{Q}$, we may define
the variational approximation as 
\begin{equation*}
\widehat{q}:=\operatornamewithlimits{argmax\,}_{q\in \mathcal{Q}}\text{ELBO}%
(q||\pi ).
\end{equation*}%
The standard approach to obtaining $\widehat{q}$ is to consider a class of
product distributions 
\begin{equation*}
\mathcal{Q}=\{q:q(x_{1}^{n},\theta )=q_{\theta }(\theta
)q_{x}(x_{1}^{n}|\theta )\},
\end{equation*}%
with $Q$ often restricted to be mean-field, i.e., $\theta _{i}$ independent
of $\theta _{j}$, $i\neq j$, and $x_{1}^{n}$ independent of $\theta $.

Regardless of the variational family adopted, $\text{KL}(q||\pi )$, and
hence $\text{ELBO}(q||\pi )$, involve both $\theta $ and $x_{1}^{n}$. {The
	product form of }{$\mathcal{Q}$}{\ allows us to write:} 
\begin{eqnarray*}
	\text{ELBO}(q||\pi ) &=&\int_{\Theta }\int_{\mathcal{X}}q_{\theta }(\theta
	)q_{x}(x_{1}^{n}|\theta )\log \frac{p(y_{1}^{n}|x_{1}^{n},\theta
		)p(x_{1}^{n}|\theta )p(\theta )}{q_{\theta }(\theta )q_{x}(x_{1}^{n}|\theta )%
	}\dt x_{1}^{n}\dt\theta \\
	&=&\int_{\Theta }\int_{\mathcal{X}}q_{\theta }(\theta
	)q_{x}(x_{1}^{n}|\theta )\log \left[ \frac{p(y_{1}^{n}|x_{1}^{n},\theta
		)p(x_{1}^{n}|\theta )}{q_{x}(x_{1}^{n}|\theta )}\right] \dt x_{1}^{n}\dt%
	\theta -\text{KL}[q_{\theta }(\theta )||p(\theta )],
\end{eqnarray*}%
{where the last line follows from Fubini's theorem and the fact that }$%
q_{x}(x_{1}^{n}|\theta )$, {by assumption, is a proper density function, for
	all }$\theta .$ {Further, defining} 
\begin{equation}
\mathcal{L}_{n}(\theta ):=\int_{\mathcal{X}}q_{x}(x_{1}^{n}|\theta )\log 
\frac{p(y_{1}^{n}|x_{1}^{n},\theta )p(x_{1}^{n}|\theta )}{%
	q_{x}(x_{1}^{n}|\theta )}\dt x_{1}^{n},  \label{l_n}
\end{equation}%
by Jensen's inequality 
\begin{flalign*}
\log p_{\theta }(y_{1}^{n}) =\log \int_{\mathcal{X}} p_{\theta
}(y_{1}^{n},x_{1}^{n})dx_{1}^{n}=\log \int_{\mathcal{X}}q_{x}(x_{1}^{n}|\theta )\left\{ \frac{%
	p(y_{1}^{n}|x_{1}^{n},\theta )p(x_{1}^{n}|\theta )}{q_{x}(x_{1}^{n}|\theta )}%
\right\} \dt x_{1}^{n}\geq \mathcal{L}_{n}(\theta ).
\end{flalign*}Thus $\mathcal{L}_{n}(\theta )$ can be viewed as an
approximation (from below) to the observed data log-likelihood. Defining 
\begin{equation}
\Upsilon _{n}(q):=\int_{\Theta }\{\log p_{\theta }(y_{1}^{n})-\mathcal{L}%
_{n}(\theta )\}q_{\theta }(\theta )\dt\theta ,  \label{Jen}
\end{equation}%
the $\text{ELBO}(q||\pi )$ can then be expressed as 
\begin{equation}
\text{ELBO}(q||\pi )=\int_{\Theta }\log p_{\theta }(y_{1}^{n})q_{\theta
}(\theta )\dt\theta -\Upsilon _{n}(q)-\text{KL}[q_{\theta }(\theta
)||p(\theta )].  \label{KL_decomp}
\end{equation}%
This representation {decomposes }$\text{ELBO}(q||\pi )${\ into three
	components}, two of which only depend on the variational approximation of
the global parameters $\theta $, and {a third component, }$\Upsilon _{n}(q)$%
, {that} \cite{yang2020alpha} refer to {as the average (with respect to }$%
q_{\theta }(\theta )$){\ \textquotedblleft Jensen's gap\textquotedblright ,
	which} encapsulates the error introduced by approximating the latent states
using a given variational class. While the first and last term in the
decomposition can easily be controlled by choosing an appropriate class for $%
q_{\theta }(\theta )$, it is the average Jensen's gap that ultimately
determines the behavior of the variational approximation.

\subsection{Consistency of variational point estimators}

The decomposition in \eqref{KL_decomp} has specific implications for
variational inference in SSMs, which can be most readily seen by first
considering the case where we only employ a variational approximation for
the states, and consider point estimation of the parameters $\theta $. In
this case, we can think of the variational family as $\mathcal{Q}%
:=\{q:q(\theta ,x_{1}^{n})=\delta _{\theta }\times q_{x}(x_{1}^{n}|\theta
)\} $, where $\delta _{\theta }$ is the Dirac delta function at $\theta $,
and we can then write 
\begin{equation*}
\frac{1}{n}\text{ELBO}(\theta \times q_{x}||\pi )=\ell _{n}(\theta )-\frac{1%
}{n}\Upsilon _{n}(\theta ,q_{x})+\frac{1}{n}\log p(\theta ),
\end{equation*}%
where we abuse notation and represent functions with arguments $\delta
_{\theta }$ only by the parameter value $\theta \in \Theta $, and also make
use of the short-hand notation $q_{x}$ for $q_{x}(x_{1}^{n}|\theta ).$
Define the variational point estimator as 
\begin{equation*}
(\widehat{\theta }_{n},\widehat{q}_{x}):=\operatornamewithlimits{argmax\,}%
_{\theta \in \Theta ,\mathcal{Q}}\frac{1}{n}\text{ELBO}(\theta \times
q_{x}||\pi ).
\end{equation*}

At a minimum, we would hope that the variational estimator $\widehat{\theta }%
_{n}$ converges to the same point as the MLE. To deduce the behavior of $%
\widehat{\theta }_{n}$, we employ the following high-level regularity
conditions.

\begin{assumption}
	\label{ass:mle} (i) The parameter space $\Theta $ is compact, and $%
	0<p(\theta)<\infty$. (ii) There exists a deterministic function $H(\theta )$%
	, continuous for all $\theta \in \Theta $, and such that $\sup_{\theta \in
		\Theta }|H(\theta )-\ell _{n}(\theta )|=o_{p}(1)$. (iii) For some value $%
	\theta_0\in\Theta$, for all $\epsilon>0$, there exists a $\delta>0$ such
	that $H(\theta _{0})\ge \sup_{d(\theta,\theta_0)>\delta}H(\theta)+\delta$.
\end{assumption}

Low level regularity conditions that imply Assumption \ref{ass:mle} are
given in \cite{douc2011consistency}. Since the main thrust{\ of this paper
	is to deduce the accuracy of }variational methods{\ in SSMs, and not }to
focus on{\ the technical details of the SSMs in particular, }we make use of
high-level conditions to simplify the exposition{\textit{\ }and reduce
	necessary technicalities that may otherwise obfuscate the main point.}

The following result shows that consistency of $\widehat{\theta }_{n}$ (for $%
\theta _{0}$) is guaranteed if the variational family for the states is
`good enough'.

\begin{lemma}
	\label{lem:first} Define $\kappa_n:=\frac{1}{n}\inf_{q_x\in\mathcal{Q}%
		_x}\Upsilon_n(\theta_0,q_x)$, and note that $\kappa_n\ge0$. If Assumption %
	\ref{ass:mle} is satisfied, and if $\kappa_n=o_p(1)$, then $\widehat{\theta}%
	_{n}\rightarrow_p\theta_0$.
\end{lemma}

The above result demonstrates that for the variational point estimator $%
\widehat{\theta }_{n}$ to be consistent, the (infeasible) average Jensen's
gap must converge to zero. Intuitively, this requires that the error
introduced by approximating the states grows more slowly than the rate at
which information accumulates in our observed sample, i.e., $n$. The
condition $\kappa _{n}=o_{p}(1)$ is stated at the true value, $\theta _{0}$,
rather than at the estimated value, as it will often be easier to deduce
satisfaction of the condition, or otherwise, at convenient points in the
parameter space.

As the following example illustrates, even in the simplest SSMs, the scaled
(average) Jensen's gap need not vanish in the limit, and can ultimately
pollute the resulting inference on $\theta _{0}$.

\begin{example}[Linear Gaussian model]
	\label{ex:one} Consider the following {SSM}, 
	\begin{flalign*}
	X_{t+1} =\rho X_{t}+\sigma_{0} \epsilon_{t}, \quad X_{1} \sim \mathcal{N}%
	\left(0, \sigma_0^2\right),\quad 
	Y_{t}=\alpha X_{t}+\sigma_{0} \eta_{t},
	\end{flalign*}with\ $\{\epsilon _{t}\}$ and $\{\eta _{t}\}$ independent
	sequences of i.i.d. standard normal random variables. We observe a sequence $%
	\{Y_{t}\}$ from the above model, but the states $\{X_{t}\}$ are unobserved.
	Furthermore, consider that $\theta =(\rho ,\alpha )^{\prime }$ {is unknown }%
	while $\sigma _{0}$ is known.
	
	We make use of the autoregressive nature of the state process to approximate
	the posterior for $\pi (x_{1}^{n}|\theta ,y_{1}^{n})$ via the variational
	family: $\mathcal{Q}_{x}:=\left\{ \lambda \in \lbrack 0,1):q_{\lambda
	}(x_{1}^{n}|\sigma _{0})=\mathcal{N}[x_{1}^{n};0,\nu (\lambda )\Phi
	_{n}(\lambda )]\right\} $, where $\nu (\lambda )=\frac{\sigma _{0}^{2}}{%
		(1-\lambda ^{2})}$ and, 
	\begin{equation*}
	\Phi _{n}(\lambda )=\left( 
	\begin{array}{ccccc}
	1 & \lambda & \lambda ^{2} & \cdots & \lambda ^{n-1} \\ 
	\lambda & 1 & \lambda & \cdots & \lambda ^{n-2} \\ 
	\lambda ^{2} & \lambda & 1 & \cdots & \lambda ^{n-3} \\ 
	\vdots & \vdots & \vdots & \ddots & \vdots \\ 
	\lambda ^{n-1} & \lambda ^{n-2} & \lambda ^{n-3} & \cdots & 1%
	\end{array}%
	\right)
	\end{equation*}%
	When evaluated at $\lambda =\rho _{0}$, $Q_{x}$ is the actual (infeasible)
	joint distribution of the states, and thus should provide a reasonable
	approximation to the state posterior.
	
	\begin{lemma}
		\label{lem:lgm} Let $\sigma_0>0$ and $0\le|\rho_0|<1$, $0\le|\alpha_0|<M$.
		Assume the variational parameter defining $\mathcal{Q}_x$ is fixed at $%
		\lambda=\rho_0$. (i) If $\rho_0=0$ and known, then the variational point
		estimator $\widehat{\alpha}$ is consistent if and only if $\alpha_0=0$. (ii)
		If $\alpha _{0}=0$ and known, then the variational point estimator for $%
		\widehat{\rho }$ is consistent if and only if $\rho _{0}=0$.
	\end{lemma}
\end{example}

Lemma \ref{lem:lgm} demonstrates that even in {this} simplest of SSMs,
variational inference is {inconsistent in anything other than }the most
vacuous cases. {In short}, so long as there is weak dependence in states the
estimator of $\alpha $ is inconsistent; alternatively, if there is no
relationship between $Y_{t}$ {and }$X_{t}$, i.e., $\alpha _{0}=0$, then the
only way in general to obtain consistent inference for $\rho _{0}$ is if $%
\rho _{0}=0$!

\subsection{Lack of Bayes consistency of the variational posterior\label%
	{vb_cons}}

While the above results pertain to variational point estimators of $\theta
_{0}$, a similar result can be stated in terms of the so-called `idealized'\
variational posterior. To state this result, we approximate the state
posterior using the class of variational approximations, 
\begin{equation*}
q_{x}(x_{1}^{n} ):=q_{\lambda }(x_{1}^{n}),
\end{equation*}%
where $\lambda \in \Lambda $ denotes the vector of {so-called `variational}
parameters' that characterize the elements in $\mathcal{Q}$. With reference
to (\ref{l_n}), making the dependence of $q_{\lambda }(x_{1}^{n})$ on the
variational parameter $\lambda $ explicit leads to the criterion $%
L_{n}(\theta ,\lambda )$, where $q_{x}(x_{1}^{n})$ in (\ref{l_n}) is
replaced by $q_{\lambda }(x_{1}^{n} )$. Optimizing over $\lambda $ for fixed 
$\theta $ yields the profiled criterion, 
\begin{equation}
\widehat{L}_{n}(\theta ):={\mathcal{L}}_{n}[\theta ,\widehat{\lambda }%
_{n}(\theta )]\equiv \sup_{\lambda \in \Lambda }\mathcal{L}_{n}(\theta
,\lambda ),  \label{l_hat}
\end{equation}%
and the `idealized'\ variational posterior for $\theta $, 
\begin{equation*}
\widehat{q}(\theta |y_{1}^{n})\propto \exp \left\{ \widehat{{L}}_{n}(\theta
)\right\} p(\theta ).
\end{equation*}%
We remark that, unlike with the frequentist optimization problem, the
idealized VB posterior incorporates {a component of }Jensen's gap directly
into the definition of that posterior. A sufficient condition for the `VB
ideal' to concentrate onto $\theta _{0}$ is that $\theta _{0}$ is the {maximum%
} of a well-defined limit counterpart to $\widehat{L}_{n}(\theta )$.
However, there is no reason to suspect this is the case a priori.

The `idealized' variational posterior $\widehat{q}(\theta |y_{1}^{n})$ is a
generalized posterior, in the sense of \cite{bissiri2016general}, based on
the profiled criterion function $\widehat{L}_{n}(\theta )$. Given that $%
\widehat{q}(\theta |y_{1}^{n})$ is constructed from a profiled criterion,
the `idealized' variational posterior is then related to the {frequentist
	profiled variational inference }approach described in \cite%
{westling2019beyond}. In their analysis, the authors view variational point
estimators of the global parameters $\theta $ as $M$-estimators based on the
profiled variational criterion function in \eqref{l_hat}. They then explore
conditions and examples under which the variational point estimator, based
on maximizing $\widehat{L}_{n}(\theta )$, do, or do not, deliver consistent
estimates of $\theta _{0}$.

While \cite{westling2019beyond} focus on consistency of variational point
estimators, we study concentration of the `idealized' posterior distribution 
$\widehat{q}(\theta |y_{1}^{n})$. The following result shows that, under
regularity conditions similar to those maintained in \cite%
{westling2019beyond}, the `idealized' variational posterior $\widehat{q}%
(\theta |y_{1}^{n})$ is Bayes consistent for some value that may or may not
coincide with $\theta _{0}$.


\begin{assumption}
	\label{ass:post} (i) There exists a map $\theta\mapsto
	\lambda(\theta)\in\Lambda$ such that $\sup_{\theta\in\Theta}\|\widehat{%
		\lambda}_n(\theta)-\lambda(\theta)\|=o_p(1) $. (ii) There exist a
	deterministic function $\mathcal{L}:\Theta\times\Lambda\mapsto\mathbb{R}$
	and a $\theta_\star\in\Theta$ such that the following are satisfied: (a) for
	all $\epsilon>0$ there exists some $\delta>0$ such that $\inf_{\theta\in
		d(\theta,\theta_\star)>\epsilon}\left[{\mathcal{L}}(\theta,\lambda(\theta))-%
	\mathcal{L}(\theta_\star,\lambda(\theta_\star))\right]\le-\delta$; (b) $%
	\sup_{\theta\in\Theta,\lambda\in\Lambda}|{\mathcal{L}}_n(\theta,\lambda)/n-%
	\mathcal{L}(\theta,\lambda)|=o_p(1)$. (iii) For any $\epsilon >0$, $%
	\int_{\Theta }\mathbf{1}\left\{ \theta :\mathcal{L}(\theta ,\lambda )-%
	\mathcal{L}(\theta _{\star },\lambda (\theta _{\star }))<\epsilon \right\}
	p(\theta )\dt\theta >0$. (iv) For all $n$ large, $\int_{\Theta }\exp \left\{ 
	\widehat{L}_{n}(\theta )\right\} p(\theta )\dt\theta <\infty $.
\end{assumption}

\begin{lemma}
	\label{lem:second} Under Assumption \ref{ass:post}, for any $\epsilon>0$, $%
	\widehat{Q}\left(\{\theta\in\Theta:d(\theta,\theta_\star)>\epsilon\}|y_1^n
	\right)=o_p(1).$
\end{lemma}

Assumption \ref{ass:post}(2.b) implies that $L_{n}(\theta ,\lambda )/n$
converges to $\mathcal{L}(\theta ,\lambda) $ (uniformly in $\theta $ and $%
\lambda $); while part (2.a) is an identification condition and states that $%
\mathcal{L}(\theta ,\lambda)${\ }is maximized at some $\theta _{\star }$,
which may differ from $\theta _{0}$. This identification {condition }{makes
	clear} that if $\theta _{\star }\neq \theta _{0}$, then 
$\mathcal{L}[\theta _{\star },\lambda (\theta _{\star })]>\mathcal{L}[\theta
_{0},\lambda (\theta _{0})]$ 
and the idealized posterior for $\theta $ will not concentrate onto $\theta
_{0}$. This can be interpreted {explicitly }in terms of {Jensen's gap as
	defined in (\ref{Jen})} by recalling that under Assumption \ref{ass:mle}, $%
\ell _{n}(\theta _{0})\rightarrow _{p}H(\theta _{0})$, and by considering
the limit of (the scaled) Jensen's gap evaluated at $\theta _{0}$, 
\begin{equation*}
\operatornamewithlimits{plim\,}_{n\rightarrow \infty }\frac{1}{n}\Upsilon
_{n}\left( \theta _{0},q_{\widehat{\lambda }_{n}(\theta _{0})}\right)
=H(\theta _{0})-\mathcal{L}[\theta _{0},\lambda (\theta _{0})]\geq H(\theta
_{0})-\mathcal{L}[\theta _{\star },\lambda (\theta _{\star })]+\delta ,
\end{equation*}%
for some $\delta \geq 0$. If Assumption \ref{ass:post}(2.a) is satisfied at $%
\theta _{\star }\neq \theta _{0}$, then $\delta >0$, and $\kappa
_{n}:=\Upsilon _{n}(\theta _{0},q_{\widehat{\lambda }_{n}(\theta
	_{0})})/n\rightarrow _{p}C>0$.

Taken together, Lemmas \ref{lem:first} and \ref{lem:second} show that,
regardless of whether one conducts variational frequentist or Bayesian
inference in SSMs, consistent inference for $\theta _{0}$ will require that
a version of Jensen's gap converges to zero. Moreover, as Example \ref%
{ex:one} has demonstrated, this is not likely to occur even in simple SSMs.
The point is further exemplified in the follow example, where we explore the
Bayesian consistency of the idealized VB posterior in the same linear
Gaussian SSM.

\begin{example}[Linear Gaussian model revisited]
	Returning to the linear Gaussian SSM {in Example \ref{ex:one}}, let us {%
		again }consider the case where $\theta =(\alpha ,\rho )^{\prime }$ is {%
		unknown, while }$\sigma _{0}${\ is known}, and we consider variational
	inference for $\theta $ using the idealized variational posterior. Our
	variational family for the vector of states is again taken to be $\mathcal{Q}%
	_{x}$, which depends on the single variational parameter $\lambda $, and
	leads to a jointly Gaussian approximation with zero-mean and covariance
	matrix $\nu (\lambda )\Phi _{n}(\lambda )$ defined previously.
	
	When $\sigma^2_0=1$, and known, the limit criterion $\mathcal{L}%
	(\theta,\lambda)$ can be constructed analytically, and the mapping $%
	\theta\mapsto\lambda(\theta)$, obtained by maximizing $\mathcal{L}%
	(\theta,\lambda)$ with respect to $\lambda$ for fixed $\theta$, calculated.
	For fixed $\theta$, with $\rho\ne0$, the mapping $\theta\mapsto\lambda(%
	\theta)$ is given by (see the proof of Lemma \ref{lem:three} for details): 
	\begin{equation*}
	\lambda(\theta)=(\alpha^2 + \rho^2 - ((\alpha^2 + \rho^2 + \rho +
	1)(\alpha^2 + \rho^2 - \rho + 1))^{1/2} + 1)/\rho.
	\end{equation*}
	
	In order for the idealized variational posterior $\widehat{q}(\theta
	|y_{1}^{n})$ to concentrate onto $\theta _{0}$, we require that the limit
	maximizer of $\mathcal{L}[\theta ,\lambda (\theta )]$ coincide with $\theta
	_{0}$ (see Lemma \ref{lem:second}). The following result demonstrates that {%
		this }does not occur in general.
	
	\begin{lemma}
		\label{lem:three} Assume that $\rho \in \lbrack \underline{\rho },\overline{%
			\rho }]$, for some known $\underline{\rho }>0$, but close to zero, and $%
		\underline\rho<\overline\rho<1$, and $\alpha \in \lbrack 0,\overline{\alpha }%
		]$, for some $\overline{\alpha }>0$. Under the variational family $\mathcal{Q%
		}_{x}$, we have 
		\begin{equation*}
		\theta _{\star }=(\underline{\rho },0)^{\prime }=%
		\operatornamewithlimits{argmax\,}_{\theta \in \lbrack \underline{\rho },%
			\overline{\rho }]\times \lbrack 0,\overline{\alpha }]}\mathcal{L}[\theta
		,\lambda (\theta )].
		\end{equation*}%
		Hence, if $\rho _{0}>\underline{\rho }$, or $\alpha _{0}>0$, then $\widehat{Q%
		}(\{\theta:d(\theta ,\theta _{0})>0\}|y_{1}^{n})>0$ with probability
		converging to one.
	\end{lemma}
	
	The joint variational state approximation $q_{\lambda }(x_{1}^{n})$ produces
	a {closed-form} marginal state approximation $q_{\lambda }(x_{n})$, for any $%
	n\geq 2$. Moreover, the marginal state posterior $\pi (x_{n}|\theta
	,y_{1}^{n})$ is also {known in closed form}, for a given value of $\theta $,
	and any $n\geq 2$. Given this, we can analytically evaluate the KL
	divergence between $q_{\lambda }(x_{n})$ and $\pi (x_{n}|\theta ,y_{1}^{n})$%
	, at any $n$, to characterize the accuracy of the resulting state
	approximation.
	
	\begin{corollary}
		\label{cor:one} Under $\mathcal{Q}_{x}$, for any $%
		n\geq 2$, we have $\text{KL}[\pi (x_{n}|\theta _{0},y_{1}^{n})||q_{\lambda
			(\theta _{\star })}(x_{n})]>0.$
	\end{corollary}
	
	The above result demonstrates that for any $n\geq 2$,  the optimal variational state density is a biased
	approximation of the exact state density. Thus, even if $\theta _{0}$ were
	known, and we only wished to conduct inference on $x_{1}^{n}$, the resulting
	variational approximation of the state density would ultimately deliver a
	poor approximation.
\end{example}

\section{Implications}
The above results suggest that VB methods can lead to inaccurate inference in the case of SSMs. In this section, we discuss further the implications of {these results for inference on the global parameters, plus their implications for predictive accuracy.} 
\subsection {Inference on global parameters}
{When conducting {VB in SSMs}, the need to approximate the posterior of $x_{1}^{n}$
introduces a discrepancy between the exact posterior, and {that which results from the VB approach}}. In this way, we can view
the latent states $x_{1}^{n}$ as \textit{incidental or nuisance} parameters {%
	(see \citealp{lancaster2000incidental}, for a review)}, which are needed to
make feasible the overall optimization problem, but which, in and of
themselves, are not the object of interest. {A similar point is made by \cite{westling2019beyond} in the case of independent states, and frequentist variational inference, where the authors demonstrate that inconsistency can occur,
even in the case of independent observations, if delicate care is not taken
with the {choice of variational class for} $x_{1}^{n}.$}

{However, the incidental parameter problem} has not stopped
researchers from using VB methods to conduct inference on $\theta $ in SSMs. 
While the general conclusions elucidated above apply, in principle, to \textit{all 
	such methods}, we next discuss two specific categories of VB methods in greater detail, and comment on {their} ability to deliver
consistent {inference} for $\theta _{0}$. 

\subsubsection{Integration approaches\label{int}}

{A possible }VB approach is to first `integrate out' the latent states so
that there is no need to perform joint inference on $(\theta ,x_{1}^{n})$.
Such an approach can be motivated by the fact that if we take $%
q_{x}(x_{1}^{n}|\theta )=\pi (x_{1}^{n}|y_{1}^{n},\theta ),$ (i.e. take the
variational approximation for $x_{1}^{n}$\ to be equivalent\textbf{\ }to the
exact posterior for $x_{1}^{n}$ conditional on $\theta $),\textbf{\ }then we
can rewrite $\text{KL}(q||\pi )$ as 
\begin{eqnarray*}
	\text{KL}(q||\pi ) &=&\int_{\Theta }\int_{\mathcal{X}}q_{\theta }(\theta
	)\pi (x_{1}^{n}|y_{1}^{n},\theta )\log \frac{q_{\theta }(\theta )\pi
		(x_{1}^{n}|y_{1}^{n},\theta )}{\pi (x_{1}^{n}|y_{1}^{n},\theta )\pi (\theta
		|y_{1}^{n})}\dt x_{1}^{n}\dt\theta \\
	&=&\text{KL}[q_{\theta }||\pi (\theta |y_{1}^{n})],
\end{eqnarray*}%
with the final line exploiting the fact that $\pi
(x_{1}^{n}|y_{1}^{n},\theta )$ integrates to one for all\ $\theta .$\ Thus,
if we are able to use as our variational approximation for the states the
actual (conditional) posterior, {we can transform} a variational problem for 
$(\theta ,x_{1}^{n})$ into a variational problem for $\theta $ alone.

The above approach is adopted by \cite{loaiza2020fast}, and is applicable in
any case where draws from $p(x_{1}^{n}|y_{1}^{n},\theta )$ can be reliably
and cheaply obtained, with the resulting draws then used to `integrate out'\
the states via the above KL divergence representation. While the approach of 
\cite{loaiza2020fast} results in the above simplification, the real key to
their approach is that it can be used to unbiasedly estimate the gradient of 
$\text{ELBO}[q_{\theta }||\pi (\theta |y_{1}^{n})]$\ (equivalent, in turn,
to the gradient of the joint ELBO in (\ref{elbo}), by the above argument).
This, in turn, allows optimization over $q_{\theta }$ to produce an
approximation to the posterior $\pi (\theta |y_{1}^{n})$. Indeed, such an
approach can be applied in many SSMs, such as unobserved component models
like the LGSSM, in which draws from $\pi (x_{1}^{n}|y_{1}^{n},\theta )$ can
be generated {exactly via, for example, forward (Kalman) filtering and
	backward sampling (\citealp{carter1994gibbs,fruhwirth1994data}); or various
	nonlinear models ({e.g. those featuring }stochastic volatility), in which {%
		efficient }}Metropolis- Hastings-{{within-Gibbs algorithms are available (%
		\citealp{kim1998stochastic,jacquier2002bayesian,primiceri2005time,huber2020inducing}%
		)}. }

In cases where we are not able to sample {readily }from $\pi
(x_{1}^{n}|y_{1}^{n},\theta )$ it may still be possible to integrate out the
states {using particle filtering methods.} To this end, assume that we can
obtain an unbiased estimate of the observed data likelihood $p_{\theta
}(y_{1}^{n})$ using a particle filter, which we denote by $\widehat{p}%
_{\theta }(y_{1}^{n})$. {We follow \cite{tran2017variational}} and write $%
\widehat{p}_{\theta }(y_{1}^{n})$ as $\widehat{p}(y_{1}^{n}|\theta ,z)$ to
make the estimator's dependence on the random filtering explicit through the
dependence on a random variable $z$, with $z$ subsequently defined by the
condition $z=\log \widehat{p}(y_{1}^{n}|\theta ,z)-\log p_{\theta
}(y_{1}^{n})$. For $g(z|\theta )$ denoting the density of $z|\theta $, \cite%
{tran2017variational} consider VB for the augmented posterior 
\begin{equation*}
\pi (\theta ,z|y_{1}^{n})=\widehat{p}(y_{1}^{n}|\theta ,z)g(z|\theta
)p(\theta )/p(y_{1}^{n})=p_{\theta }(y_{1}^{n})\exp (z)g(z|\theta )p(\theta
)/p(y_{1}^{n})=\pi (\theta |y_{1}^{n})\exp (z)g(z|\theta ),
\end{equation*}%
which, marginal of $z$, has the correct target posterior $\pi (\theta
|y_{1}^{n})$ {due to the unbiasedness of the estimator }$\widehat{p}%
(y_{1}^{n}|\theta ,z)$. The {authors} refer to the resulting method as
variational Bayes with an intractable likelihood function (VBIL). The VBIL
posteriors can be obtained by considering a variational approximation to $%
\pi (\theta ,z|y_{1}^{n})$ that minimizes the KL divergence between $%
q(\theta ,z)=q_{\theta }(\theta )g(z|\theta )$ and $\pi (\theta
,z|y_{1}^{n}) $: 
\begin{flalign*}
\text{KL}[q(\theta ,z)||\pi (\theta ,z|y_{1}^{n})]&=\int_{\Theta }\int_{%
	\mathcal{Z}}q_{\theta }(\theta )g(z|\theta )\log \frac{q_{\theta }(\theta
	)g(z|\theta )}{\pi (\theta |y_{1}^{n})\exp (z)g(z|\theta )}\dt z\dt\theta  \\&=\int_{\Theta }\int_{\mathcal{Z}}q_{\theta }(\theta )g(z|\theta )\log 
\frac{q_{\theta }(\theta )}{p_{\theta }(y_{1}^{n})\exp (z)p(\theta )}\dt z\dt%
\theta +\log p(y_{1}^{n}) \\
&=-\int_{\Theta }q_{\theta }(\theta )\log p_{\theta }(y_{1}^{n})\dt\theta +%
\text{KL}(q_{\theta }||p_{\theta })+\Upsilon _{n}[q(\theta ,z)]+\log
p(y_{1}^{n}),
\end{flalign*}where in this case 
\begin{equation*}
\Upsilon _{n}[q(\theta ,z)]=\int_{\Theta }\int_{\mathcal{Z}}q_{\theta
}(\theta )g(z|\theta )\left\{ \log p_{\theta }(y_{1}^{n})-\log {\widehat{p}%
	(y_{1}^{n}|\theta ,z)}\right\} \dt z\dt\theta .
\end{equation*}%
For fixed $\theta $, $\mathbb{E}_{z}\left[ \widehat{p}(y_{1}^{n}|\theta \,z)%
\right] =p_{\theta }(y_{1}^{n})$, but in general $\log \widehat{p}%
(y_{1}^{n}|\theta ,z)$ is a biased estimator of $\log p_{\theta }(y_{1}^{n})$%
, {from which it follows that} $\Upsilon _{n}[q(\theta ,z)]\geq 0.$ However,
in contrast to the general approximation of the states discussed {in Section %
	\ref{vi}}, which intimately relies on the choice of the approximating
density $q_{x}(x_{1}^{n}|\theta )$, VBIL can achieve consistent inference on 
$\theta _{0}$ by choosing {an appropriate} number of particles $N$ {in the
	production of }$\widehat{p}(y_{1}^{n}|\theta ,z)$ .

To see this, we recall that a maintained assumption in the literature on {%
	PMCMC} methods is that, for all $n$ and $N$, the conditional mean and
variance of the density $g(z|\theta )$ satisfy $\mathbb{E}[z|\theta
]=-\gamma (\theta )^{2}/2N$, and $\text{Var}\left[ z|\theta \right] =\gamma
(\theta )^{2}/N$, where $\gamma (\theta )^{2}$ is bounded uniformly over $%
\Theta $; see, e.g., Assumption 1 in \cite{doucet2015efficient} and
Assumption 1 in \cite{tran2017variational}. However, in general, $N$ is
assumed to be chosen so that $\mathbb{E}[z|\theta ]=-\sigma ^{2}/2$ and $%
\text{Var}\left[ z|\theta \right] =\sigma ^{2}>0$, $0<\sigma <\infty $. Note
that, under this choice for $N$, for any $\varepsilon >0$ 
\begin{equation*}
\lim_{n\rightarrow \infty }\text{Pr}\left[ \Upsilon _{n}[q(\theta
_{0},z)]/n>\varepsilon \right] =\lim_{n\rightarrow \infty }\text{Pr}\left[
-q_{\theta }(\theta _{0})\mathbb{E}\left[ z|\theta _{0}\right] >n\epsilon %
\right] =\lim_{n\rightarrow \infty }\text{Pr}\left[ q_{\theta }(\theta
_{0})\sigma ^{2}/2>n\varepsilon \right] =0,
\end{equation*}%
assuming $q_{\theta }(\theta _{0}),\sigma ^{2}<\infty $.

From this condition, we see that the VBIL inference problem is
asymptotically the same as the VB inference problem for $\theta $ alone.
Consequently, existing results on the posterior concentration of VB methods
for $\theta $ alone can be used to deduce posterior concentration of the
VBIL posterior for $\theta $.

\subsubsection{Structured approximations of the states}

Yet {another approach} for dealing with variational inference {in the
	presence of} states is to consider a structured approximation that allows
for a dynamic updating of the approximation for the posterior of the states.
Such an approximation can be achieved by embedding in the class of
variational densities an analytical filter, like the Kalman filter. \cite%
{koop2018variational} propose the use of the Kalman filter \textit{within VB}
{(VBKF) }as a means of approximating the posterior density of the states
using Kalman recursions. In particular, the authors approximate the
posterior $\pi (x_{1}^{n}|y_{1}^{n},\theta )$ by approximating the
relationship between $X_{t}$ and $X_{t-1}$, which may in truth be non-linear
in $\theta $, by the random walk model $X_{t}=X_{t-1}+\epsilon _{t}$, with $%
\epsilon _{t}\sim i.i.d.N(0,\sigma _{0}^{2})$, and then use Kalman filtering
to update the states in conjunction with a linear approximation to the
measurement equation. Using this formulation, the variational approximation
is of the form $q(x_{1}^{n},\theta )=q_{\theta }(\theta )q_{x}(x_{1}^{n})$,
where $q_{x}(x_{1}^{n})\propto \prod_{k\geq 1}\exp (-\left\{ x_{k}-\widehat{x%
}_{k|k}\right\} ^{2}(1-\mathcal{K}_{k})P_{k|k-1}/2)$ and where the terms $%
\mathcal{K}_{k},P_{k|k-1},\widehat{x}_{k|k}$ are explicitly calculated using
the Kalman recursion: $\widehat{x}_{k|k}=\widehat{x}_{k|k-1}+\mathcal{K}%
_{k}(y_{t}^{\star }-\widehat{x}_{k|k-1}),$ and where $\mathcal{K}_{k}$ is
the Kalman gain, $P_{k|k-1}$ is the predicted variance of the state, and in
the application of \cite{koop2018variational}, $y_{t}^{\star }=\log
(y_{t}^{2})$.

While the solution proposed by the VBKF is likely to lead to better
inference on the states, especially when $x_{1}^{n}$ behaves like a random
walk, ultimately we are still `conducting inference'\ on $x_{1}^{n}$, and
thus we still encounter the incidental parameter problem {as a consequence}.
Indeed, taking as the variational family for $x_{1}^{n}$ the Kalman filter
approximation yields, at time $k\geq 1$, a conditionally normal density with
mean $\hat{x}_{k|k}$ and variance $(1-\mathcal{K}_{k})P_{k|k-1}$. Hence, we
have a variational density that has the same structure as in Lemma \ref%
{lem:lgm}, but which allows for a time varying mean and variance. Given this
similarity, there is no reason to suspect that such an approach will yield
inferences that are consistent. Indeed, further intuition can be obtained by
noting that, in the VBKF formulation, the simplification of the state
equation means that we disregard any dependence between the states and the
values of $\theta $ that drive their dynamics.

The variational approach of \cite{chan2020fast} can be viewed similarly: the
suggested algorithm assumes and exploits a particular dynamic structure for
the states that allows for analytical (posterior) updates and thus leads to
computationally simple estimates for the variational densities of $%
q_{x}(x_{1}^{n}).$ As with the VBKF approach, the assumed nature of the
state process used by \cite{chan2020fast} to estimate $q_{x}(x_{1}^{n})$
implies that, in general, it is unlikely that Bayesian consistency can be
achieved. {Due to space restrictions, further discussion on the specifics of this
approach are relegated to Sections \ref{approx} and \ref{sec:cy2020} of the Supplementary Appendix}.

\subsection{VB-based prediction\label%
	{pred}}

{VB provides, at best, an approximation to
the posterior and, as a result, may well yield less accurate inferences than
those produced by the exact posterior (see, e.g. \citealp{koop2018variational}; %
\citealp{gunawan2020variational}). However,  VB} \textit{can} perform admirably in predictive settings{, see, e.g., {\cite%
		{quiroz2018gaussian}\ and \cite{frazierloss}, amongst others}, in
	the sense of replicating the out-of-sample accuracy achieved by exact
	predictives, when such comparators are available. }(See 
\citealp{frazier2019approximate} {for a comparable finding in the context of
	predictions based on }ABC.) Therefore, even though the VB posterior may not
necessarily converge to the true value $\theta _{0}$, so long as the value onto which
it is concentrating is not too far away from $\theta_0$, it may be that VB-based predictions perform well in practice.

Recall the conditional density of $Y_{n+1}$ given $x_{n+1}$ and $\theta $ is 
$g_{\theta }(Y_{n+1}|x_{n+1})$, so that the predictive pdf for $Y_{n+1}$ can
be expressed as 
\begin{flalign}
p(Y_{n+1}|y_{1}^{n}) &=\int_{\Theta }\int_{\mathcal{X}}g_{\theta }(Y_{n+1}|x_{n+1})\pi (x_{1}^{n+1},\theta |y_{1}^{n})%
\dt x_{1}^{n+1}\dt\theta \label{eq:pred}\\
&=\int_{\Theta }\int_{\mathcal{X}}\int_{\mathcal{X}}g_{\theta
}(Y_{n+1}|x_{n+1})\underbrace{p(x_{n+1}|x_{n},y_{1}^{n},\theta
	)p(x_{1}^{n}|y_{1}^{n},\theta )}_{(1)}\underbrace{\pi (\theta |y_{1}^{n})}%
_{(2)}\dt x_{n+1}\dt x_{1}^{n}\dt\theta ,  \nonumber
\end{flalign}
where the last line follows from the Markovianity of the state
transition equation (see equation \eqref{eq:HMM_1}). In many large SSMs,
using MCMC methods to estimate \eqref{eq:pred} is infeasible or {prohibitive
	computationally}, {due to the difficulty of sampling from }$\pi
(x_{1}^{n+1},\theta |y_{1}^{n}).$ {Instead}, VB methods can produce an
estimate of $p(Y_{n+1}|y_{1}^{n})$ by approximating, in various
ways, the two pieces in equation \eqref{eq:pred} underlined as (1) and (2).
All such methods replace the second underlined term by some approximate
posterior for $\theta $, but differ in how they access the first underlined
term.

In all the cases {of which }we are aware, we can separate VB methods
for prediction in SSMs into two {classes}: a class which makes explicit use
of a variational approximation to the states, $\widehat{q}_{x}$ {to replace }%
$p(x_{1}^{n}|y_{1}^{n},\theta )$; and a class that uses an accurate
simulation-based estimate of $p(x_{1}^{n}|y_{1}^{n},\theta )$. Due to space restrictions, we do not give a detailed discussion of how these VB predictives are produced, and instead refer the interested reader to Section \ref{sec:vbpred} in the Supplementary Appendix. 

Any Bayesian method that replaces $\pi(\theta|\y)$ in part (2) by an approximation, e.g., $\widehat{q}_\theta$ in the case of VB, will lead to some inaccuracy,  however, as shown by \cite{frazier2019approximate} in the case of ABC, this loss in accuracy is often minimal. Therefore, what really matters in terms of accurate prediction in SSMs using VB is the replacement of (1) in \eqref{eq:pred}. Replacing (1) with an accurate simulation-based estimate is likely to deliver more accurate estimators, at the cost of additional computation. However, it is not necessarily clear that the resulting predictions will perform much better than those approaches based on the approximation $\widehat{q}_x$. In the following section, we demonstrate that even through the inference that results from using $\widehat{q}_x$ instead of (1) can be poor, the resulting predictive performance is often quite reasonable, at least for sample sizes that are not too large.

\section{Numerical assessment of VB methods}

In this section, we shed {further }light on the phenomenon of {the }%
predictive accuracy of VB methods, and connect the performance of these
methods to the inconsistency for $\theta _{0}$ that can result as the sample
size diverges. {The results suggest that, in terms of predictive accuracy, there is little difference
between methods  in small
sample sizes or with a small number of out-of-sample observations. However, we document a
clear hierarchy across methods as the sample}
size becomes larger and as the out-of-sample evaluation increases.

\subsection{Simulation design}

{We now compare the inferential and
predictive accuracy of the variational methods of \cite%
{quiroz2018gaussian} and \cite{loaiza2020fast} against an exact MCMC-based
estimate of $\pi (\theta ,x_{1}^{n}|y_{1}^{n})$, referred to as `exact
Bayes' hereafter,  in a simulation exercise}. {In Section \ref{sec:compdets} of the Supplementary Appendix, we provide complete details on the implementation of each of these methods under this particular simulation design.} However, we remark here that \cite{quiroz2018gaussian} is an example of a VB method in which the states are approximated via a particular choice of variational family, whilst \cite{loaiza2020fast} (as noted in Section \ref{int}) adopt a variational approximation for the posterior of the global parameters only, with the conditional posterior of the states accessed via simulation.

The assumed DGP is specified as an unobserved component
model with stochastic volatility (UCSV): 
\begin{equation}
\mu _{t}=\bar{\mu}+\rho _{\mu }\left( \mu _{t-1}-\bar{\mu}\right) +\sigma
_{\mu }\varepsilon _{t},\;\;h_{t}=\bar{h}+\rho _{h}\left( h_{t-1}-\bar{h}%
\right) +\sigma _{h}\eta _{t},\;\;Y_{t}=\mu _{t}+\exp (h_{t}/2)u_{t},\;\;
\label{1}
\end{equation}%
where $(\varepsilon _{t},\eta _{t},u_{t})^{^{\prime }}\overset{i.i.d.}{\sim }%
N(0,I_{3})$. The unobserved component term $\mu _{t}$ is a latent variable
that captures the persistence in the conditional mean of $Y_{t}$, while the
stochastic volatility term $h_{t}$ captures the persistence in the
conditional variance. We consider the following three set of values for the
true parameters: 
\begin{equation*}
\begin{array}{cc}
\text{DGP 1:} & \bar{\mu}_{0}=0;\text{ }{\rho _{\mu }}_{0}=0.8;\text{ }{%
	\sigma _{\mu }}_{0}=0.5;\text{ }\bar{h}_{0}=-1.0;\text{ }{\rho _{h}}%
_{0}=0.00;\text{ }{\sigma _{h}}_{0}=0.0 \\ 
\text{DGP 2:} & \bar{\mu}_{0}=0;\text{ }{\rho _{\mu }}_{0}=0.0;\text{ }{%
	\sigma _{\mu }}_{0}=0.5;\text{ }\bar{h}_{0}=-1.3;\text{ }{\rho _{h}}%
_{0}=0.95;\text{ }{\sigma _{h}}_{0}=0.3 \\ 
\text{DGP 3:} & \bar{\mu}_{0}=0;\text{ }{\rho _{\mu }}_{0}=0.8;\text{ }{%
	\sigma _{\mu }}_{0}=0.5;\text{ }\bar{h}_{0}=-1.3;\text{ }{\rho _{h}}%
_{0}=0.95;\text{ }{\sigma _{h}}_{0}=0.3%
\end{array}%
\end{equation*}%
The specifications for DGP 1 produce a time series process that has
substantial persistence in the conditional mean, and a constant variance;
DGP 2 generates a process that has substantial persistence in the
conditional variance, and a fixed marginal mean of zero; whilst DGP 3 corresponds to
a process that exhibits persistence in both the conditional mean and
variance. The true parameter vector in each case is defined as $\theta
_{0}=\left( \bar{\mu}_{0},{\rho _{\mu }}_{0},{\sigma _{\mu }}_{0},\bar{h}%
_{0},{\rho _{h}}_{0},{\sigma _{h}}_{0},\right) ^{\prime }$.

For the predictive assessment we compare exact Bayes with the two
variational methods cited above plus the method of \cite{chan2020fast}. As
discussed in Section \ref{sec:compdets} of the Supplementary Appendix, the method of \cite{chan2020fast} exploits a very specific structure in the
construction of the variational algorithm, which in this case
corresponds to DGP 2 under the parameter restrictions ${\rho _{h}}_{0}=1.0$, ${\sigma_{\mu }}_0=0.0$, and $\bar{h}_{0}=0.0$. Thus, application of this approach under any of the above true
DGPs constitutes misspecified inference; hence, we do not include this
technique in the inferential assessment. Due to space constraints, certain tables and figures are included in Section \ref{app:results} of the Supplementary Appendix.

\subsection{Accuracy of inference on the states}

\label{sec: stat_inf}

We assess inferential accuracy through lens of state estimation. To this end,
we generate a times series of length $T=11000${\ from each of the three true
	DGP specifications.} {The full sample} is used {to produce the} exact
posterior as well as {the two} {approximate posteriors corresponding to the
	QNK and LSND methods};{\ hence, we are able to shed some light on the
	theoretical consistency results provided above. }We assess the inferential
accuracy of each method {(exact and approximate) }by calculating the root
mean squared error ({RMSE)} and mean absolute error (MAE) of {each sequence
	of} {marginal }posterior means, {for }$t=1,2,....,T$, for the unobserved
component, $\mu _{t},$ and the stochastic {{standard deviation}, $\exp
	(h_{t}/2)$}, relative to the{\ marginal posterior means that }results when
we condition on the true parameters, denoted respectively by $\mathbb{E}[\mu
_{t}|\theta _{0},y_{1}^{T}]$ and $\mathbb{E}[\exp (h_{t}/2)|\theta
_{0},y_{1}^{T}]$, $t=1,2,....,T.$ The results are presented in Table~\ref{tab:inferacc}.

As expected, exact Bayes produces the most accurate point estimates for the
	two sets of latent variables (both $\mu _t$ and $\exp (h_t/2)$), as
tallies with the theoretical guarantees of this method. In terms of the VB methods, the LSND results closely match those of exact
	Bayes as this method does not suffer from the incidental parameter problem. In contrast, the QNK method
does not deal directly with this problem and, as a consequence, exhibits -
across \textit{all of the designs} recorded in Table \ref{tab:inferacc} -
inaccuracy that is between two and ten times greater than that of both exact Bayes and the LSND method. From the results recorded in Table \ref{tab:comp} in Supplementary Appendix \ref{app:results}, we also note that the time taken to estimate the UCSV model, under all three DGPs, is approximately the same for exact Bayes and the LSND method, with the QNK approach taking roughly twice as long as both.

\begin{table}[h]
	\caption{Accuracy in the estimation of the unobserved component and
		conditional standard deviation. Panel A presents the root mean squared error
		(RMSE) and mean absolute error (MAE) of the posterior mean estimates of the
		unobserved component ($\protect\mu _{t}$). The columns correspond to the
		three DGP specifications, while the rows correspond to the three predictive
		methods: exact Bayes, LSND and QNK. Panel B presents the corresponding
		results for the posterior mean estimates of the conditional standard
		deviation. The unobserved component error measures are computed relative to $%
		\mathbb{E}[\protect\mu _{t}|\protect\theta _{0},y_{1}^{T}]$, while the
		conditional standard deviation error measures are computed relative to $%
		\mathbb{E}[\exp (h_{t}/2)|\protect\theta _{0},y_{1}^{T}]$, where $\theta _{0}$ denotes the true parameter vector and $T=11000$.}
	\label{tab:inferacc}\centering{\footnotesize
		\begin{tabular}{lrrrrlrrr}
			\hline\hline
\textbf{Panel A: ($\mu _{t}$)}			&        &                   {RMSE} &        &  &             &        &                   {MAE} &        \\
			&  DGP 1 &                    DGP 2 &  DGP 3 &  &             &  DGP 1 &                   DGP 2 &  DGP 3 \\ \cline{2-4}\cline{7-9}
			Exact Bayes & 0.0463 &                   0.0207 & 0.0203 &  & Exact Bayes & 0.0382 &                  0.0004 & 0.0155 \\
			LSND        & 0.0495 &                   0.0253 & 0.0271 &  & LSND        & 0.0405 &                  0.0006 & 0.0207 \\
			QNK         & 0.1211 &                   0.2646 & 0.1098 &  & QNK         & 0.0980 &                  0.0700 & 0.0664 \\ \hline
	\textbf{Panel B: ($\exp(h_{t}/2) $)}				&        & \multicolumn{1}{c}{RMSE} &        &  &             &        & \multicolumn{1}{c}{MAE} &        \\
			&  DGP 1 &                    DGP 2 &  DGP 3 &  &             &  DGP 1 &                   DGP 2 &  DGP 3 \\ \cline{2-4}\cline{7-9}
			Exact Bayes & 0.0497 &                   0.0255 & 0.0234 &  & Exact Bayes & 0.0470 &                  0.0198 & 0.0180 \\
			LSND        & 0.0520 &                   0.0309 & 0.0315 &  & LSND        & 0.0490 &                  0.0242 & 0.0246 \\
			QNK         & 0.0984 &                   0.2505 & 0.2231 &  & QNK         & 0.0983 &                  0.2179 & 0.1656 \\ \hline\hline
		\end{tabular}%
	}
\end{table}

We further highlight the results in Table \ref{tab:inferacc} by plotting, in
Figure \ref{fig:states_post_mean}, the marginal posterior means for both $%
\mu _{t}${\ and }$\exp (h_{t}/2)${, for each point in time across a given
	sample period, for all three methods; with the sequence of `true' posterior
	means (that condition on }$\theta _{0}$){\ included for comparison.} For the
sake of brevity, we only present results for DGP 2, with the corresponding
results for DGPs 1 and 3 placed in Supplementary Appendix \ref{app:results}. Consistent with the summary results in Table \ref{tab:inferacc}, the
posterior means for exact Bayes and LSND are both very similar, for each $t$%
, and visually very close to the corresponding true time posterior means,
across the entire sample. In comparison, the QNK method consistently
produces {point estimates of }the states that are {very different }from the {%
	values that condition on the true parameters}, {as accords with the
	dependence of the method on a variational approximation for the states, and
	the consequent loss of Bayesian consistency for }$\theta _{0}$. We note that
the additional figures in Supplementary Appendix \ref{app:results}
demonstrate that, at least visually, the QNK method seems to produce more
accurate estimates of $\mu _{t}$ under DGPs 1 and 3 than it does
	under DGP 2; however, it remains inaccurate in terms of estimating 
$\exp (h_{t}/2)$ under these alternative DGPs.

\begin{figure}[h!]
	\begin{center}
			\includegraphics*[scale = 0.5]{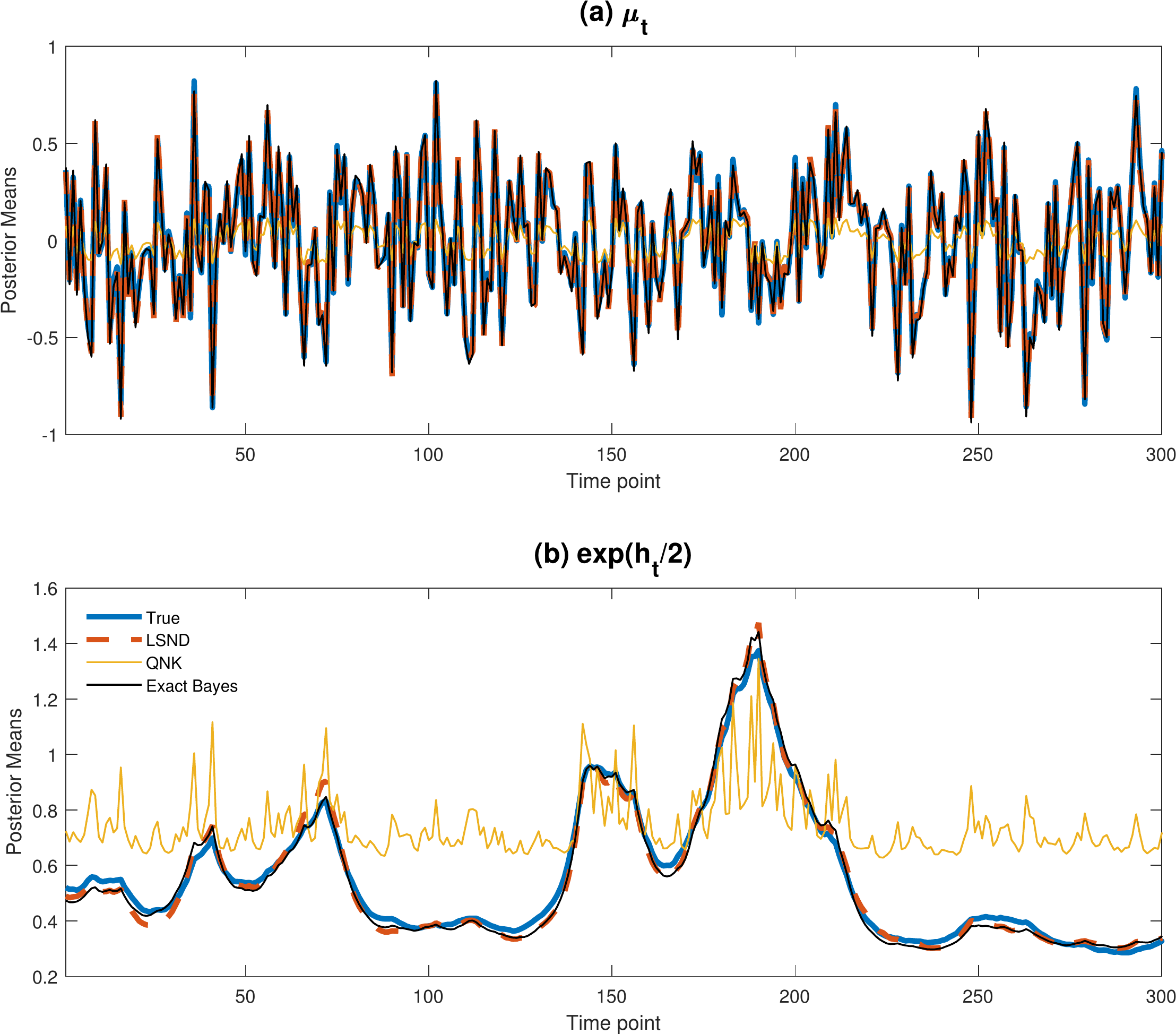}
	\end{center}
	\caption{Posterior means of the latent states, under DGP 2, over the first
		300 time points. Panel (a) plots the posterior mean for $\protect\mu_t$. The
		red, gold and black lines plot, respectively, the posterior means based on
		the LSND, QNK and exact Bayes approaches. The blue line plots the posterior
		mean that conditions on the true parameters. Panel (b) presents
		corresponding results for $\exp(h_t/2)$.}
	\label{fig:states_post_mean}
\end{figure}


\subsection{Predictive accuracy}\label{sec:comp}%

To assess the predictive accuracy of each method we conduct an expanding
window prediction exercise using the same generated data as in the previous
subsection. The exercise consists of constructing the Bayesian predictive
density for $Y_{n+1}$, conditional on the sample $y_{1}^{n}$, for each of
the competing approaches and for $n\in \{1000,\dots ,T-1\}$. For each method
and each out-of-sample time point we evaluate eight measures of predictive
accuracy: the logarithmic score, four censored scores, the continuously
ranked probability score, the tail weighted continuously ranked probability
score and the interval score. Details of all scoring rules, including
appropriate references, are provided in Section \ref{sect:scoring rules} of
the Supplementary Appendix. We document results using 100, 1000
	and 10000 out-of-sample evaluations respectively, remembering that the CY
method is now included in the comparison, but only for the case of DGP 2. For reasons of space, we only present results for the largest number of out-of-sample evaluations (10000) in the main text, in Table \ref{tab:predacc}, while the results for the other evaluation periods are given in Section \ref{app:results} of the Supplementary Appendix, in Tables \ref{tab:predacc100} and \ref{tab:predacc1000} respectively.

Focussing first on the results in Table~\ref{tab:predacc},
	based on the very large number of out-of-sample evaluations, we observe an
	interesting ranking. Across all designs, and according to all measures of
	accuracy, exact Bayes is the most accurate method. As accords with the
	inferential results discussed above, the LSND method has a predictive
	accuracy that often matches, or is extremely similar to, that of exact
	Bayes, followed, in order, by CY and QNK.  A similar ranking holds
	for the results recorded in {Tables \ref{tab:predacc100} and \ref{tab:predacc1000}} in Section \ref{app:results} of the Supplementary Appendix. However, the differences {between methods }are
somewhat less stark over the smaller out-of-sample evaluation periods, which
highlights the fact that it is ultimately the consistency properties of
the different VB methods (in evidence for the largest evaluation period, given the large size of the expanding estimation windows)
that is driving the discrepancies between the predictive accuracy of the
competing methods.

Whilst a ranking is {in evidence} {in Table \ref{tab:predacc}}, it
can be argued that across \textit{certain} DGP and scoring rule
combinations, the predictive results across the different methods are 
{still quite} similar, both between the exact and (all) VB methods, and
between the different VB methods. That is, for certain combinations of DGPs
and scoring rules, all methods are {seen to perform well (relative to
	the benchmark of the true predictive)}, and the more substantial \textit{%
	inferential} discrepancies observed between certain of the methods are not
reflected at the predictive level. This finding corroborates the point made
earlier, and which has been supported by other findings in the literature,
namely that computing a posterior {via an approximate method does not 
	\textit{necessarily} reduce predictive accuracy (relative to exact Bayes) by
	a }substantial{\textit{\ }amount. }

{However, despite there being} certain DGP and scoring rule combinations
where the methods perform similarly, this is not true across all DGPs and
loss measures, {in particular for the larger out-of-sample evaluation
	period}. For example, {and with specific reference to Table \ref%
	{tab:predacc},} there is a clear trend that as model complexity increases 
{(i.e. moving from DGP 1 through to DGP 3)}, variational methods that
work harder to correctly approximate the states have greater predictive
accuracy. This finding is particularly marked for the log score and the
interval score, which directly measure the dispersion of the posterior
predictive. {In the case of DGP 3}, the all-purpose variational
method of \cite{quiroz2018gaussian} performs the worst across all the 
{methods} under analysis, {and most notably for the log score
	and the interval score}. This feature is most likely due to the fact that
the posteriors associated with the method of \cite{quiroz2018gaussian} have
overly thin tails. Consequently, parameter uncertainty is not adequately
accounted for when constructing the posterior predictive, {which} results in
a predictive with thin tails, and ultimately translates into poor
performance in scores that measure both location and/or dispersion.

\begin{table}[h!]
	\caption{Predictive performance of competing Bayesian approaches: exact
		Bayes, LSND and CY and QNK. The column labels indicate the out-of-sample
		predictive performance measure while the row labels indicate the predictive
		method. `True DGP' indicates the productive results that condition on the
		true parameters. Panels A, B and C correspond to the results for DGP 1, 2
		and 3, respectively. The average predictive measures in this table were
		computed using 10000 out-of-sample evaluations.}
	\label{tab:predacc}\centering{\footnotesize
		\begin{tabular}{lcccccccc}
			\hline\hline
				{\underline{\textbf{Panel A: DGP 1}}}	& {LS} & {CS-10\%} & {CS-20\%} & {CS-80\%} & {CS-90\%} & {CRPS} & {TWCRPS }
			& {IS} \\ \cline{2-9}
			True DGP & -1.259 & -0.308 & -0.508 & -0.505 & -0.297 & -0.481 & -0.146 & 
			-4.001 \\ 
			Exact Bayes & -1.260 & -0.308 & -0.508 & -0.506 & -0.297 & -0.481 & -0.147 & 
			-4.012 \\ 
			LSND & -1.261 & -0.308 & -0.509 & -0.507 & -0.298 & -0.481 & -0.147 & -4.015
			\\ 
			CY & - & - & - & - & - & - & - & - \\ 
			QNK & -1.262 & -0.309 & -0.509 & -0.507 & -0.298 & -0.482 & -0.147 & -4.030
			\\ \hline
	{\underline{\textbf{Panel B: DGP 2}}}		& {LS} & {CS-10\%} & {CS-20\%} & {CS-80\%} & {CS-90\%} & {CRPS} & {TWCRPS }
			& {IS} \\ \cline{2-9}
			True DGP & -1.192 & -0.341 & -0.551 & -0.551 & -0.340 & -0.454 & -0.139 & 
			-4.097 \\ 
			Exact Bayes & -1.193 & -0.342 & -0.551 & -0.551 & -0.340 & -0.454 & -0.139 & 
			-4.093 \\ 
			LSND & -1.194 & -0.342 & -0.551 & -0.552 & -0.341 & -0.454 & -0.139 & -4.101
			\\ 
			CY & -1.205 & -0.346 & -0.557 & -0.555 & -0.344 & -0.456 & -0.140 & -4.182
			\\ 
			QNK & -1.212 & -0.350 & -0.560 & -0.560 & -0.349 & -0.456 & -0.140 & -4.316
			\\ \hline
{\underline{\textbf{Panel C: DGP 3}}}			& {LS} & {CS-10\%} & {CS-20\%} & {CS-80\%} & {CS-90\%} & {CRPS} & {TWCRPS }
			& {IS} \\ \cline{2-9}
			True DGP & -1.268 & -0.304 & -0.505 & -0.521 & -0.300 & -0.490 & -0.150 & 
			-4.424 \\ 
			Exact Bayes & -1.268 & -0.305 & -0.506 & -0.520 & -0.299 & -0.491 & -0.150 & 
			-4.423 \\ 
			LSND & -1.271 & -0.306 & -0.507 & -0.521 & -0.301 & -0.491 & -0.150 & -4.442
			\\ 
			CY & - & - & - & - & - & - & - & - \\ 
			QNK & -1.301 & -0.315 & -0.521 & -0.536 & -0.311 & -0.497 & -0.152 & -4.707
			\\ \hline\hline
		\end{tabular}%
	}
\end{table}

\section{Discussion}

We have systematically documented the behavior of variational methods, in
terms of inference and prediction, within the class of state space models 
(SSMs). Sufficient conditions for (both frequentist and Bayesian)
consistency of variational inference (VI) in SSMs have been presented in
terms of the so-called Jensen's gap, which measures the discrepancy
introduced within VI due to the approximation of the states. {Focusing on }%
variational {Bayes (VB) methods specifically, we show that only} methods that are
capable of closing Jensen's gap yield Bayesian consistent inference for the global parameters and, in turn, deliver more accurate inferences for the states.

In the context of empirically relevant SSMs, we find numerical evidence of a clear hierarchy in terms of the accuracy of state inference across different variational methods: methods that can close Jensen's gap produce
qualitatively more accurate inferences than those that do not. However, whilst this same hierarchy also holds for VB-based prediction, we find that the extent to which different variational approaches vary in terms for predictive accuracy depends on the data generating process (DGP), the loss in which
the different methods are evaluated, {and - most importantly - the size of the out-of-sample
	evaluation period.} Indeed, we document that there are certain
circumstances, i.e., {sample size, }DGP and loss combinations, where there is little to separate the various approaches. However, in large samples, methods that attain Bayesian consistent inference on the global parameters produce more accurate predictions.

To keep the length of this paper manageable, we have deliberately analysed
and compared only a select few of the variational methods used to conduct
inference and prediction in SSMs. Our findings, however, suggest that
certain classes of approximations for the state posterior employed in the
machine learning literature, e.g., classes based on normalising or
autoregressive flows, may be flexible enough to deliver accurate inferences
and predictions; we refer to, e.g., \citealp{ryder2018black}, and the
references therein, for a discussion of such methods in SSMs. We leave a
comparison between the approaches discussed herein and those commonly used
in machine learning for future research.

\bibliographystyle{Myapalike}
\bibliography{Prediction_bib_merged_clean}

\appendix

\section{Further details and discussion on variational methods in SSMs}
\subsection{Methods for producing variational predictives}\label{sec:vbpred}
Following on from the discussion in Section \ref{pred} in the main text, in this section we give precise details on how the variational predictives are constructed. 
Recall that the predictive pdf for $Y_{n+1}$ can
be expressed as 
\begin{flalign}
p(Y_{n+1}|y_{1}^{n}) &=\int_{\Theta }\int_{\mathcal{X}}g_{\theta }(Y_{n+1}|x_{n+1})\pi (x_{1}^{n+1},\theta |y_{1}^{n})%
\dt x_{1}^{n+1}\dt\theta \label{eq:pred}\\
&=\int_{\Theta }\int_{\mathcal{X}}\int_{\mathcal{X}}g_{\theta
}(Y_{n+1}|x_{n+1})\underbrace{p(x_{n+1}|x_{n},y_{1}^{n},\theta
	)p(x_{1}^{n}|y_{1}^{n},\theta )}_{(1)}\underbrace{\pi (\theta |y_{1}^{n})}%
_{(2)}\dt x_{n+1}\dt x_{1}^{n}\dt\theta ,  \nonumber
\end{flalign}
VB methods can produce an
estimate of $p(Y_{n+1}|y_{1}^{n})$ by\textit{\ } approximating, in various
ways, the two pieces in equation \eqref{eq:pred} underlined as (1) and (2). VB methods
for prediction in SSMs either make explicit use
of a variational approximation to the states, $\widehat{q}_{x}$ {to replace }%
$p(x_{1}^{n}|y_{1}^{n},\theta )$; or use an accurate
simulation-based estimate of $p(x_{1}^{n}|y_{1}^{n},\theta )$. We now discuss these two approaches in more detail. 
\subsubsection{Approximation approaches\label{approx}}

{The VB methods that approximate }$p(Y_{n+1}|y_{1}^{n})${\ by constructing
	an approximation to} $p(x_{n+1}|x_{n},y_{1}^{n},\theta )$ $\times
p(x_{1}^{n}|y_{1}^{n},\theta )$ {all make} {use of a variational
	approximation }$\widehat{q}_{x}$ of $p(x_{1}^{n}|y_{1}^{n},\theta )$, {in
	addition }to using the structure of the state equation. To illustrate this,
it is perhaps easiest to consider the case where we seek to estimate %
\eqref{eq:pred} by generating values of $Y_{n+1}$ and using as our estimate
of $p(Y_{n+1}|y_{1}^{n})$ the kernel density obtained from the simulations.
In this way, we can see that simulation of $Y_{n+1}$ requires {simulating
	the }following random variables, {in sequence:} 
\begin{equation*}
\theta |y_{1}^{n};\;x_{1}^{n}|y_{1}^{n},\theta
;\;x_{n+1}|x_{n},y_{1}^{n},\theta ;\text{ and }Y_{n+1}|x_{n+1},\theta .
\end{equation*}%
More precisely, consider a fixed value of $\theta ^{(j)}$ drawn from some
variational approximation of $\pi (\theta |y_{1}^{n})$, call it $\widehat{q}%
_{\theta }$. Given the realization $\theta ^{(j)}$, we simulate $%
x_{1}^{n}|y_{1}^{n},\theta ^{(j)}$ from the VB approximation of the states $%
\widehat{q}_{x}$. Next, given $x_{n}^{(j)}\sim \widehat{q}_{x}$, we can
generate $x_{n+1}^{(j)}$ from $p(x_{n+1}|x_{n}^{(j)},y_{1}^{n},\theta
^{(j)}) $ by generating from the transition density of the states, $%
x_{n+1}^{(j)}\sim \chi _{\theta }(x_{n+1},x_{n}^{(j)})$, and under the draws 
$x_{n}^{(j)}$ and $\theta ^{(j)}$. Lastly, $Y_{n+1}^{(j)}$ is generated
according to the conditional distribution $Y_{n+1}^{(j)}\sim g_{\theta
}(y_{n+1}|x_{n+1}^{(j)})$. While the above steps are simple to implement,
the critical point to realize is that since $x_{n}^{(j)}$ has not been
generated from $p(x_{1}^{n}|y_{1}^{n},\theta )$, in general $x_{n+1}^{(j)}$
is not a draw from $p(x_{n+1}|x_{n},y_{1}^{n},\theta
)p(x_{1}^{n}|y_{1}^{n},\theta )$. Hence, the draw $Y_{n+1}^{(j)}$ does not
correctly reflect the structure of the assumed model, and $Y_{n+1}^{(j)}$
cannot be viewed as being a draw from the exact predictive density in %
\eqref{eq:pred}.

Notable uses of the above approach to prediction appear in \cite%
{quiroz2018gaussian}, \cite{koop2018variational} and \cite{chan2020fast}.
While similar in form and structure, {these }three specific approaches are
distinct in the sense that the each use different methods to construct $%
\widehat{q}_{x}$ ({in addition to the differences in the construction of }$%
\widehat{q}_{\theta }$) and thus {to generate} $x_{n+1}^{(j)}$.

\subsubsection{Simulation approaches\label{sim}}

As an alternative, one may {estimate }$p(Y_{n+1}|y_{1}^{n})$ using {exact }%
draws of $Y_{n+1}$, $x_{n+1}$ and $x_{1}^{n}$, conditional on the draw of $%
\theta $ from some $\widehat{q}_{\theta }.$ {For example, if draws from the
	exact posterior of the states, }$p(x_{1}^{n}|y_{1}^{n},\theta )$, {are
	readily available via an efficient MCMC algorithm, }$p(Y_{n+1}|y_{1}^{n})$
can {be estimated via the same set of steps as delineated above, apart from }%
$x_{n}^{(j)}$ {being drawn directly from }$p(x_{1}^{n}|y_{1}^{n},\theta )$, {%
	rather than some }$\widehat{q}_{x}${; see, for example, \cite{loaiza2020fast}%
	. In this case, }$x_{n+1}^{(j)}$ {is a draw from }$%
p(x_{n+1}|x_{n},y_{1}^{n},\theta )p(x_{1}^{n}|y_{1}^{n},\theta )$ and,
consequently, the draw $Y_{n+1}^{(j)}$ correctly reflects the model
structure. Moreover, due to the Markovian nature of \eqref{eq:HMM_1},
posterior {draws of the full vector of states }$x_{1}^{n}$ {are not
	required, only draws of }$x_{n}.$ {As such,} any forward (particle)
filtering method is all that is required to produce draws of $x_{n}$ that
are conditional on the full vector of observations.
\section{Computational details and additional results: numerical exercise}
\subsection{Computational details: methods}\label{sec:compdets}
This section contains detailed discussions on the computational methods used in the numerical examples in Section \ref{vi} of the main paper.
\subsubsection{Exact Bayes}

Denote the two vectors of latent variables as $\mu _{1}^{n}=\left( \mu
_{1},\dots ,\mu _{n}\right) ^{\prime }$ and $h_{1}^{n}=\left( h_{1},\dots
,h_{n}\right) ^{\prime }$. The exact posterior density is given as 
\begin{equation}
\pi (\theta ,\mu _{1}^{n},h_{1}^{n}|y_{1}^{n})=\frac{p(y_{1}^{n}|\mu
	_{1}^{n},h_{1}^{n},\theta )p(\mu _{1}^{n},h_{1}^{n}|\theta )p(\theta )}{%
	p(y_{1}^{n})},  \label{eq:postsucsv}
\end{equation}%
with prior $p(\theta )=p(\bar{\mu})p({\rho _{\mu }})p({\sigma _{\mu }})p(%
\bar{h})p({\rho _{h}})p({\sigma _{h}})$, where $\bar{\mu}\sim N(0,1000)$, ${%
	\rho _{\mu }}\sim \text{U}(0,1)$, ${\sigma _{\mu }}^{2}\sim \text{IG}%
(1.001,1.001)$, $\bar{h}\sim N(0,1000)$, ${\rho _{h}}\sim \text{U}(0,1)$ and 
${\sigma _{h}}^{2}\sim \text{IG}(1.001,1.001)$. We draw from %
\eqref{eq:postsucsv} using an MCMC algorithm. Specifically, the vector $%
h_{1}^{n}$ is generated using the method proposed in \cite{primiceri2005time}%
, while $\mu _{1}^{n}$ is generated using the forward-filtering
backward-sampling method in \cite{carter1994gibbs}. Given the choice of
priors, the parameters $\bar{\mu}$, ${\sigma _{\mu }}$, $\bar{h}$ and ${%
	\sigma _{h}}$ can be generated directly using Gibbs steps. The parameters ${%
	\rho _{\mu }}$ and ${\rho _{h}}$ are generated using a Metropolis-Hastings
step with a Gaussian proposal distribution. The corresponding predictive
(expressed using obvious notation), 
\begin{equation}
p(Y_{n+1}|y_{1}^{n})=\int_{\Theta }\int_{\mathcal{H}}\int_{\mathcal{M}%
}g_{\theta }(Y_{n+1}|\mu _{n+1},h_{n+1})p(\mu _{n+1},h_{n+1}|\theta ,\mu
_{1}^{n},h_{1}^{n})\pi (\theta ,\mu _{1}^{n},h_{1}^{n}|y_{1}^{n})\dt\mu
_{1}^{n}\dt h_{1}^{n}\dt\theta ,  \label{exact_pred}
\end{equation}%
is then estimated (via kernel density methods) using the draws of $Y_{n+1}$
obtained conditional on the draws of $\theta ,$ $\mu _{1}^{n}$ and $%
h_{1}^{n}.$

\subsubsection{Quiroz et al. (2018)}

Re-cast in terms of our simulation design, \cite{quiroz2018gaussian} (QNK
hereafter) adopt the variational approximation: 
\begin{equation}
q_{\widehat{\lambda }}(\theta ,\mu _{1}^{n},h_{1}^{n})=q_{\widehat{\lambda }%
	_{1}}(\theta )q_{\widehat{\lambda }_{2}}(x_{1}^{n}),  \label{eq:gfb}
\end{equation}%
where $\widehat{\lambda }=\left( \widehat{\lambda }_{1}^{\prime },\widehat{%
	\lambda }_{2}^{\prime }\right) ^{\prime }$, $x_{t}=\left( \mu
_{t},h_{t}\right) ^{\prime }$ and $x_{1}^{n}=\left( x_{1}^{^{\prime }},\dots
,x_{n}^{^{\prime }}\right) ^{^{\prime }}$. The approximations $q_{\widehat{%
		\lambda }_{1}}(\theta )$ and $q_{\widehat{\lambda }_{2}}(x_{1}^{n})$ are
optimal elements in the variational classes $\mathcal{Q}_{1}=\{q_{\lambda
	_{1}}(\theta ):\lambda _{1}\in \Lambda _{1}\}$ and $\mathcal{Q}%
_{2}=\{q_{\lambda _{2}}(\theta ):\lambda _{2}\in \Lambda _{2}\}$,
respectively, where the optimization is performed using a stochastic
gradient ascent (SGA) algorithm \citep{bottou2010large}, and the
approximation is based on the same prior as specified above. The elements of
the first class are Gaussian densities of the form $q_{\lambda _{1}}(\theta
)=\phi _{6}\left( \theta ;\nu _{\theta },BB^{^{\prime }}+\text{diag}%
(d^{2})\right) $, while the elements of the second class are of the form $%
q_{\lambda _{2}}(x_{1}^{n})=$ $\phi _{2n}\left( x_{1}^{n};\nu
_{x},CC^{^{\prime }}\right) $, where $C$ is a three diagonal lower
triangular matrix, and the subscript on the symbol for the normal pdf, $\phi 
$, denotes the dimension of the density. (For more details on this
approximating class see \citealp{ong2018gaussian}.) Replacing $\pi (\theta
,\mu _{1}^{n},h_{1}^{n}|y_{1}^{n})$ in (\ref{exact_pred}) by the
approximation\ in (\ref{eq:gfb}), the predictive density is then estimated
as described in Section \ref{approx}.

\subsubsection{Loaiza-Maya et al. (2021)}

Once again translating their method into our setting, \cite{loaiza2020fast}
(LSND hereafter), in contrast to \cite{quiroz2018gaussian}, adopt a
variational approximation for $\pi (\theta |y_{1}^{n})$ only, exploiting the
exact conditional posterior density of the states, $p\left( \mu
_{1}^{n},h_{1}^{n}|y_{1}^{n},\theta \right) .$\ As such, the variational
approximation takes the form: 
\begin{equation}
q_{\widehat{\lambda }}(\theta ,\mu _{1}^{n},h_{1}^{n}|y_{1}^{n})=q_{\widehat{%
		\lambda }}(\theta )p\left( \mu _{1}^{n},h_{1}^{n}|y_{1}^{n},\theta \right) ,
\label{eq:hvbsucsv}
\end{equation}%
where $q_{\widehat{\lambda }}(\theta )$ is an optimal element in the
variational class $\mathcal{Q}=\{q_{\lambda }(\theta ):\lambda \in \Lambda
\} $, once again found via SGA. For $\mathcal{Q}$ the class of multivariate
Gaussian densities with a factor structure is employed, so that $q_{\lambda
}(\theta )=\phi _{6}\left( \theta ;\nu ,BB^{^{\prime }}+\text{diag}%
(d^{2})\right) $, and $\lambda =\left( \nu ^{\prime },\text{ vec}(B)^{\prime
},d^{^{\prime }}\right) ^{\prime }$. Replacing $\pi (\theta ,\mu
_{1}^{n},h_{1}^{n}|y_{1}^{n})$ in (\ref{exact_pred}) by the approximation in
(\ref{eq:hvbsucsv}) (once again, with the same underlying prior adopted),
the predictive density is then estimated as described in Section \ref{sim}.
Generation from $p\left( \mu _{1}^{n},h_{1}^{n}|y_{1}^{n},\theta \right) $
is achieved via an MCMC algorithm that sequentially draws from $p\left( \mu
_{1}^{n}|y_{1}^{n},\theta ,h_{1}^{n}\right) $ using the method in \cite%
{carter1994gibbs}; and then draws from $p\left( h_{1}^{n}|y_{1}^{n},\theta
,\mu _{1}^{n}\right) $ using the approach in \cite{primiceri2005time}.

\subsubsection{Chan and Yu (2020)}\label{sec:cy2020}

The final VB method we consider is that of \cite{chan2020fast} (CY
hereafter). This approach has been designed specifically for (vector)
autoregressive models with stochastic volatility (SV) and not for the UCSV
model in (\ref{1}).\ The SV component(s) is (are) assumed to have random
walk dynamics, which are factored into the construction of the VB
approximation for the states. In the case of a scalar random variable (and
volatility state) the assumed structure is:%
\begin{equation*}
h_{t}=h_{t-1}+\sigma _{h}\eta _{t},\;\;Y_{t}=\exp (h_{t}/2)u_{t}.
\end{equation*}%
Denoting by $h_{0}$ the initial condition of the states, and defining $%
\theta =(\sigma _{h},h_{0})^{\prime }$, CY construct an approximation to the
exact posterior $p(\theta ,h_{1}^{n}|y_{1}^{n})$ as: 
\begin{equation*}
q_{\widehat{\lambda }}(\theta ,h_{1}^{n})=q_{\widehat{\lambda }_{1}}(\sigma
_{h}^{2})q_{\widehat{\lambda }_{2}}(h_{0})q_{\widehat{\lambda }%
	_{3}}(h_{1}^{n}),
\end{equation*}%
where $q_{\widehat{\lambda _{1}}}(\sigma _{h}^{2})$, $q_{\widehat{\lambda
		_{2}}}(h_{0})$ and $q_{\widehat{\lambda _{3}}}(h_{1}^{n})$ are optimal
elements in the variational classes $\mathcal{Q}_{1}=\{q_{\lambda
	_{1}}(\sigma _{h}^{2}):\lambda _{1}\in \Lambda _{1}\}$, $\mathcal{Q}%
_{2}=\{q_{\lambda _{2}}(h_{0}):\lambda _{2}\in \Lambda _{2}\}$ and $\mathcal{%
	Q}_{3}=\{q_{\lambda _{3}}(h_{1}^{n}):\lambda _{3}\in \Lambda _{3}\}$,
respectively. The elements of each class are defined respectively as $%
q_{\lambda _{1}}(\sigma _{h}^{2})=\mathcal{IG}(\sigma _{h}^{2};\nu ,S)$, $%
q_{\lambda _{2}}(h_{0})=\phi _{1}(h_{0};\mu _{0},s_{0}^{2})$ and $q_{\lambda
	_{3}}(h_{1}^{n})=\phi _{n}(h_{1}^{n};m,\hat{K}^{-1})$. The variational
parameters $\lambda _{1}=\left( \nu ,S\right) ^{^{\prime }}$, $\lambda
_{2}=\left( \mu _{0},s_{0}^{2}\right) ^{^{\prime }}$ and $\lambda _{3}=m$,
are calibrated to produce the elements in $\mathcal{Q}_{1}$, $\mathcal{Q}%
_{2} $ and $\mathcal{Q}_{3}$ that minimise the KL divergence from $p(\theta
,h_{1}^{n}|y_{1}^{n})$. The authors use a coordinate ascent algorithm %
\citep{blei2017variational} to perform the optimization, while the value of $%
\hat{K}^{-1}$ can be optimally computed as a deterministic function of $%
\lambda _{1}$, $\lambda _{2}$, $\lambda _{3}$ and $y_{1}^{n}$. In our
implementation of the CY method, the priors are set to $p(\sigma _{h}^{2})=%
\mathcal{IG}(\sigma _{h}^{2};1.001,1.001)$ and $p(h_{0})=\phi _{1}\left(
h_{0};0,1000\right) $, where $\mathcal{IG}$ denotes the inverse gamma
distribution. The predictive density is then estimated as described in
Section \ref{approx}, with $h_{1}^{n}$\ playing the role of $x_{1}^{n}$
therein.

\subsection{Additional numerical results}\label{app:results}

\subsubsection{State inference}

This section contains additional details for the inferential state
comparison given in Section \ref{sec: stat_inf}. Firstly, the computation
times for the different VB methods used in this section are given in Table %
\ref{tab:comp}, and demonstrate that exact Bayes and the method of \cite%
{loaiza2020fast} are comparable in terms of computational cost across the
different simulation designs. In contrast, the method of \cite%
{quiroz2018gaussian} takes roughly twice as long to implement. 
\begin{table}[!htbp]
	\caption{Estimation time (in seconds) required to estimate the UCSV model on
		a sample of 11000 time points.}
	\label{tab:comp}\centering{\footnotesize \ 
		\begin{tabular}{lrrr}
			\hline\hline
			& \multicolumn{3}{c}{Estimation times in seconds} \\ 
			&  &  &  \\ 
			& DGP 1 & DGP 2 & DGP 3 \\ \cline{2-4}
			&  &  &  \\ 
			Exact Bayes & 69.1643 & 66.0900 & 71.9203 \\ 
			LSND & 67.4664 & 68.5716 & 71.6248 \\ 
			QNK & 125.3298 & 130.4489 & 120.8658 \\ \hline\hline
		\end{tabular}
	}
\end{table}

Figures \ref{fig:states_post_meanDGP1} and \ref{fig:states_post_meanDGP3}
plot the posteriors for the unknown states under DGP 1 and 3; see Section %
\ref{sec: stat_inf} for details regarding the production and interpretation
of these plots. Comparable to the results under DGP 2 in Figure \ref%
{fig:states_post_mean}, we see that the method of \cite{quiroz2018gaussian}
performs the worst in terms of state inference for $\exp (h_{t}/2)$ across
both DGPs, while the method of \cite{loaiza2020fast} performs similarly to
exact Bayes. For inference on the time-varying mean, $\mu _{t}$, all methods
perform better in general than under DGP 2; in particular, the method of 
\cite{quiroz2018gaussian} {appears (visually) to produce} much more
accurate inferences under DGPs 1 and 3\textbf{\ }than in the case of DGP2. 
\begin{figure}[!htbp]
	\begin{center}
		\includegraphics*[scale = 0.7]{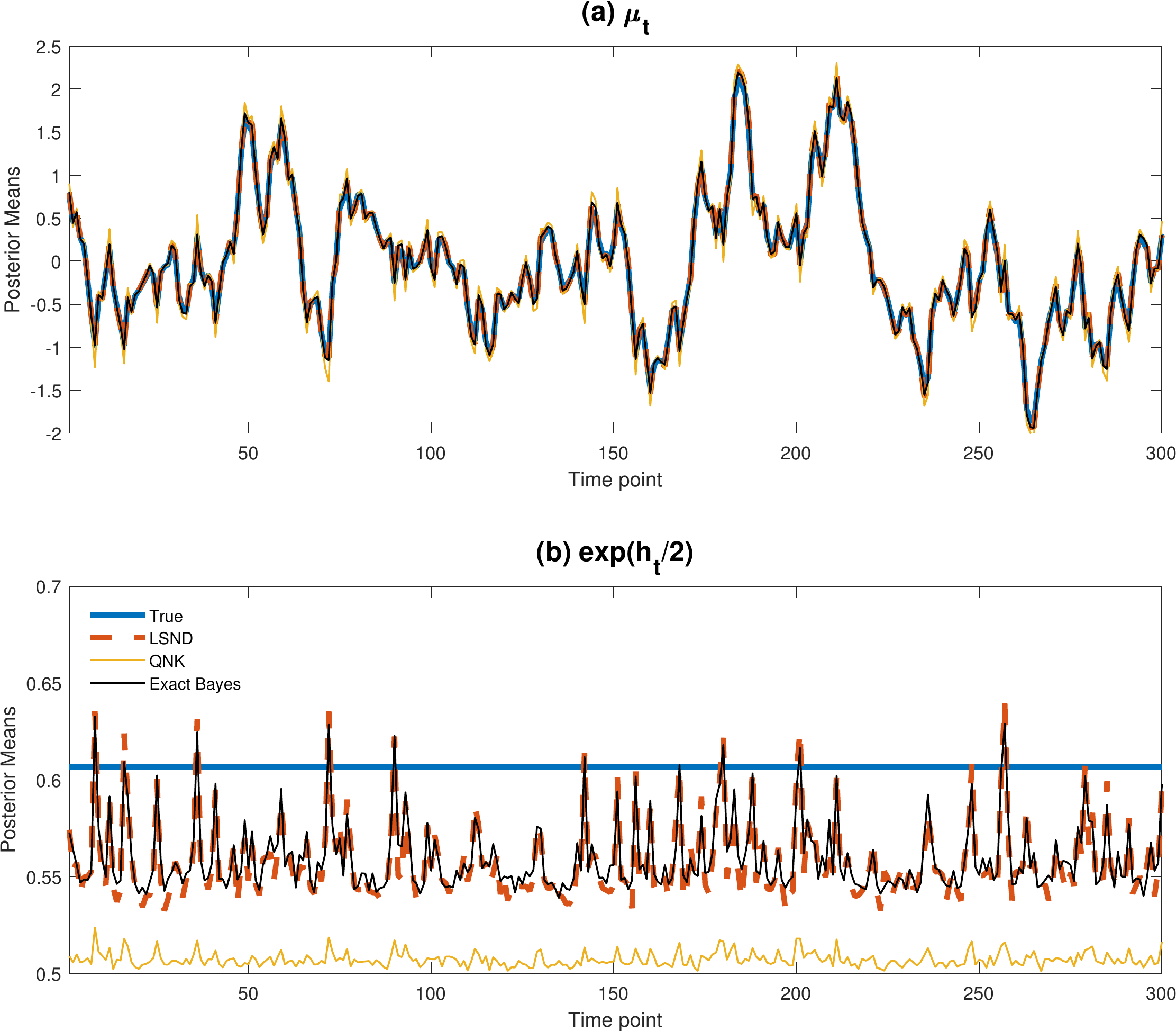}
	\end{center}
	\caption{Posterior means of the latent states, under DGP 1, over the first
		300 time points. Panel (a) plots the posterior mean for $\protect\mu _{t}$.
		The red, gold and black lines plot, respectively, the posterior means based
		on the LSND, QNK and exact Bayes approaches. The blue line plots the
		posterior mean that conditions on the true parameters. In the case of the
		conditional standard deviation, $\exp (h_{t}/2)$, conditional on $\protect%
		\theta _{0}$, $\exp (h_{t}/2)$ is constant, and, thus, so is the posterior
		mean. Panel (b) presents corresponding results for $\exp (h_{t}/2)$.}
	\label{fig:states_post_meanDGP1}
\end{figure}

\begin{figure}[!htbp]
	\begin{center}
		\includegraphics*[scale = 0.7]{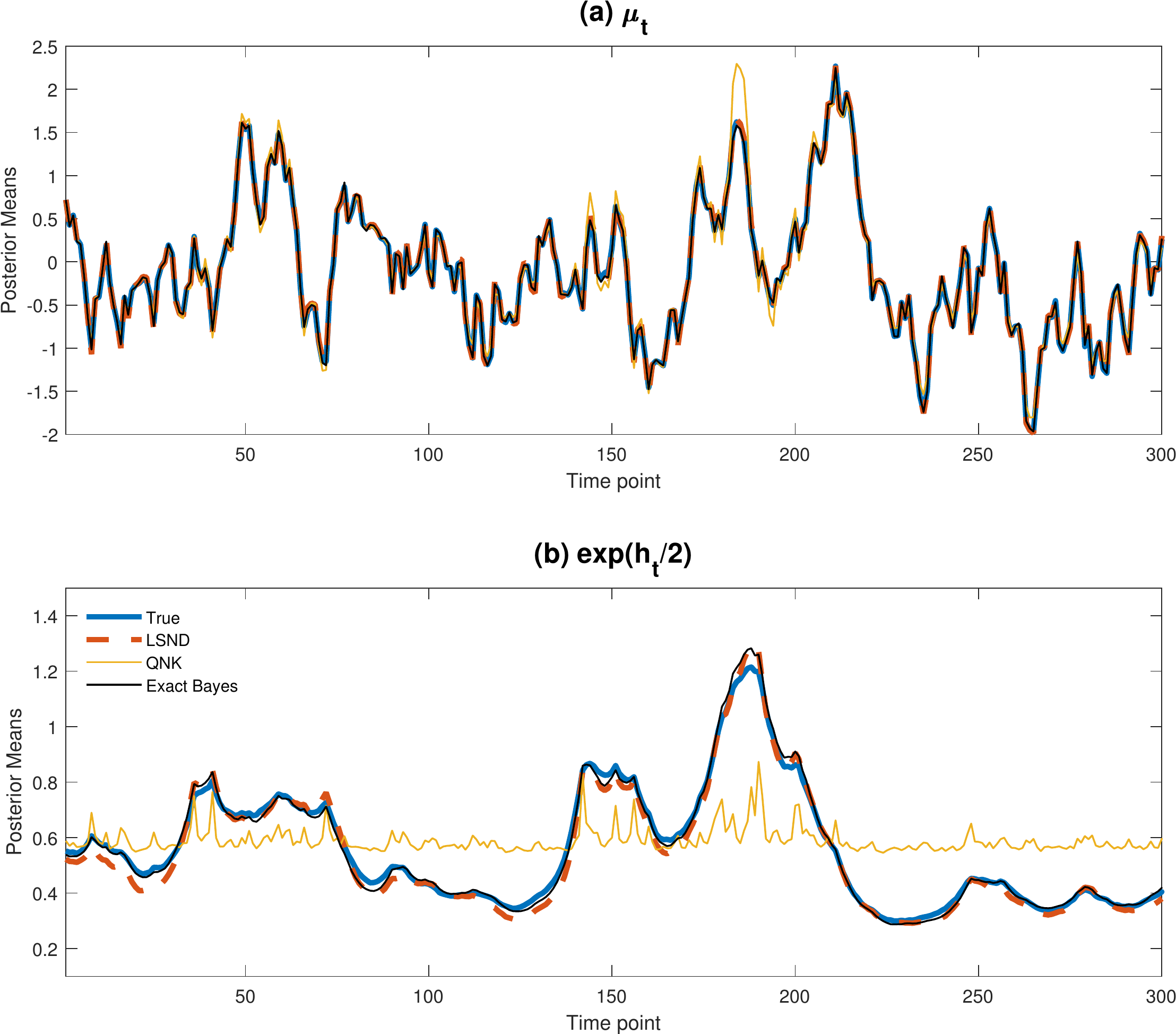}
	\end{center}
	\caption{Posterior means of the latent states, under DGP 3, over the first
		300 time points. Panel (a) plots the posterior mean for $\protect\mu _{t}$.
		The red, gold and black lines plot, respectively, the posterior means based
		on the LSND, QNK and exact Bayes approaches. The blue line plots the
		posterior mean that conditions on the true parameters. Panel (b) presents
		corresponding results for $\exp (h_{t}/2)$.}
	\label{fig:states_post_meanDGP3}
\end{figure}

\subsubsection{Predictive performance}
Herein, we present the results of our predictive analysis for the additional samples sizes referenced in Section \ref{sec:comp} of the main paper. Table \ref{tab:predacc100} contains results for 100 out-of-sample evaluations and Table \ref{tab:predacc1000} results for 1000 out-of-sample evaluations.

\begin{table}[!htbp]
	\caption{Predictive performance of competing Bayesian approaches: exact
		Bayes, LSND and CY and QNK. The column labels indicate the out-of-sample
		predictive performance measure while the row labels indicate the predictive
		method. `True DGP' indicates the productive results that condition on the
		true parameters. Panels A, B and C correspond to the results for DGP 1, 2
		and 3, respectively. The average predictive measures in this table were
		computed using 100 out-of-sample evaluations.}\label{tab:predacc100}\centering{\footnotesize
		\begin{tabular}{lcccccccc}
			\hline\hline
			\textbf{Panel A: DGP{\ 1}}		& {LS} & {CLS-10\%} & {CLS-20\%} & {CLS-80\%} & {CLS-90\%} & {CRPS} & {\
				TWCRPS } & {MSIS} \\ \cline{2-9}
			True DGP & -1.206 & -0.410 & -0.647 & -0.348 & -0.166 & -0.453 & -0.134 & 
			-3.849 \\ 
			Exact Bayes & -1.210 & -0.416 & -0.651 & -0.349 & -0.166 & -0.456 & -0.136 & 
			-3.894 \\ 
			LSND & -1.210 & -0.415 & -0.648 & -0.352 & -0.170 & -0.455 & -0.135 & -3.929
			\\ 
			CY & - & - & - & - & - & - & - & - \\ 
			QNK & -1.212 & -0.419 & -0.649 & -0.354 & -0.170 & -0.454 & -0.136 & -3.999
			\\ \hline
			{\textbf{Panel B: DGP{\ 2}}} 		& {LS} & {CLS-10\%} & {CLS-20\%} & {CLS-80\%} & {CLS-90\%} & {CRPS} & {\
				TWCRPS } & {MSIS} \\ \cline{2-9}
			True DGP & -1.132 & -0.295 & -0.488 & -0.568 & -0.332 & -0.428 & -0.128 & 
			-4.417 \\ 
			Exact Bayes & -1.129 & -0.300 & -0.495 & -0.551 & -0.315 & -0.429 & -0.128 & 
			-4.339 \\ 
			LSND & -1.144 & -0.294 & -0.489 & -0.570 & -0.333 & -0.430 & -0.128 & -4.250
			\\ 
			CY & -1.159 & -0.314 & -0.506 & -0.564 & -0.321 & -0.432 & -0.130 & -4.583
			\\ 
			QNK & -1.151 & -0.311 & -0.504 & -0.562 & -0.322 & -0.432 & -0.129 & -4.762
			\\ \hline
			{\textbf{Panel C: DGP{\ 2}}} 			& {LS} & {CLS-10\%} & {CLS-20\%} & {CLS-80\%} & {CLS-90\%} & {CRPS} & {\
				TWCRPS } & {MSIS} \\ \cline{2-9}
			True DGP & -1.182 & -0.391 & -0.569 & -0.365 & -0.197 & -0.452 & -0.132 & 
			-4.665 \\ 
			Exact Bayes & -1.185 & -0.401 & -0.576 & -0.351 & -0.187 & -0.451 & -0.133 & 
			-4.632 \\ 
			LSND & -1.188 & -0.399 & -0.574 & -0.359 & -0.197 & -0.452 & -0.133 & -4.634
			\\ 
			CY & - & - & - & - & - & - & - & - \\ 
			QNK & -1.218 & -0.415 & -0.605 & -0.368 & -0.201 & -0.455 & -0.135 & -5.119
			\\ \hline\hline
		\end{tabular}%
	}
\end{table}

\begin{table}[!htbp]
	\caption{Predictive performance of competing Bayesian approaches: exact
		Bayes, LSND and CY and QNK. The column labels indicate the out-of-sample
		predictive performance measure while the row labels indicate the predictive
		method. `True DGP' indicates the productive results that condition on the
		true parameters. Panels A, B and C correspond to the results for DGP 1, 2
		and 3, respectively. The average predictive measures in this table were
		computed using 1000 out-of-sample evaluations.}\label{tab:predacc1000}\centering{\footnotesize
		\begin{tabular}{lcccccccc}
			\hline\hline
			{\textbf{Panel A: DGP{\ 1}}} 			& {LS} & {CLS-10\%} & {CLS-20\%} & {CLS-80\%} & {CLS-90\%} & {CRPS} & {\
				TWCRPS } & {MSIS} \\ \cline{2-9}
			True DGP & -1.243 & -0.320 & -0.524 & -0.495 & -0.273 & -0.473 & -0.143 & 
			-3.928 \\ 
			Exact Bayes & -1.246 & -0.319 & -0.524 & -0.497 & -0.273 & -0.474 & -0.143 & 
			-3.935 \\ 
			LSND & -1.247 & -0.319 & -0.522 & -0.498 & -0.274 & -0.474 & -0.143 & -3.942
			\\ 
			CY & - & - & - & - & - & - & - & - \\ 
			QNK & -1.246 & -0.319 & -0.523 & -0.498 & -0.273 & -0.474 & -0.143 & -3.934
			\\ \hline
			{\textbf{Panel B: DGP{\ 2}}} 			& {LS} & {CLS-10\%} & {CLS-20\%} & {CLS-80\%} & {CLS-90\%} & {CRPS} & {\
				TWCRPS } & {MSIS} \\ \cline{2-9}
			True DGP & -1.261 & -0.379 & -0.595 & -0.593 & -0.403 & -0.482 & -0.146 & 
			-4.477 \\ 
			Exact Bayes & -1.268 & -0.382 & -0.600 & -0.593 & -0.405 & -0.483 & -0.146 & 
			-4.504 \\ 
			LSND & -1.273 & -0.382 & -0.600 & -0.598 & -0.409 & -0.484 & -0.146 & -4.531
			\\ 
			CY & -1.284 & -0.386 & -0.604 & -0.601 & -0.410 & -0.486 & -0.147 & -4.623
			\\ 
			QNK & -1.274 & -0.385 & -0.602 & -0.596 & -0.409 & -0.484 & -0.146 & -4.597
			\\ \hline
			{\textbf{Panel C: DGP{\ 3}}} 			& {LS} & {CLS-10\%} & {CLS-20\%} & {CLS-80\%} & {CLS-90\%} & {CRPS} & {\
				TWCRPS } & {MSIS} \\ \cline{2-9}
			True DGP & -1.329 & -0.324 & -0.548 & -0.558 & -0.313 & -0.517 & -0.156 & 
			-4.752 \\ 
			Exact Bayes & -1.334 & -0.327 & -0.551 & -0.557 & -0.312 & -0.518 & -0.157 & 
			-4.746 \\ 
			LSND & -1.338 & -0.328 & -0.552 & -0.559 & -0.314 & -0.518 & -0.157 & -4.773
			\\ 
			CY & - & - & - & - & - & - & - & - \\ 
			QNK & -1.345 & -0.332 & -0.559 & -0.566 & -0.322 & -0.519 & -0.157 & -4.987
			\\ \hline\hline
		\end{tabular}%
	}
\end{table}

\subsection{Additional details: scoring rules}

\label{sect:scoring rules} In the simulation exercises we have considered
five different forms of positively-oriented scoring rules to measure
predictive accuracy. To express each of these scoring rules, denote as $%
P\left( Y_{n+1}|y_{1}^{n}\right) $ the predictive distribution associated
with the Bayesian predictive density $p(Y_{n+1}|y_{1}^{n})$.

The first scoring rule that we consider is the logarithmic score (LS), which
is given {by} 
\begin{equation}
S_{\text{LS}}(P\left( Y_{n+1}|y_{1}^{n}\right) ,y_{n+1})=\ln p\left(
y_{n+1}|y_{1}^{n}\right) .  \label{ls_prelim}
\end{equation}%
This score is favourable to predictive distributions that assign high\textbf{%
	\ }probability mass to the realised value $y_{n+1}$.

The second type of scoring rule that we consider is the censored logarithm
score (CS) introduced by \cite{diks2011likelihood}. This rule is defined as 
\begin{equation}
S_{\text{CS}}(P\left( Y_{n+1}|y_{1}^{n}\right) ,y_{n+1})=\ln p\left(
y_{n+1}|y_{1}^{n}\right) I\left( y_{n+1}\in A\right) +\left[ \ln
\int_{A^{c}}p\left( y|y_{1}^{n}\right) dy\right] I\left( y_{n+1}\in
A^{c}\right) .  \label{csr_prelim}
\end{equation}%
This score rewards predictive accuracy over the region of interest $A$ ({with%
} $A^{c}${\ indicating the complement of this region). Here we report
	results solely for }$A$ {defining} the lower {and upper tail of the
	predictive distribution, as determined respectively by the }10\%, 20\%, 80\%
and 90\% quantiles{\ of the empirical distribution of }$y_{t}${. We label
	these scores as {CS-10\%}, {CS-20\%}, {CS-80\%} and CS-90\% }.

The third scoring rule is the continuously ranked probability score (CRPS)
proposed by \cite{gneiting2007strictly} and defined as%
\begin{equation}
S_{\text{CRPS}}\left[ P(Y_{n+1}|y_{1}^{n}),y_{n+1}\right] =-\int_{-\infty
}^{\infty }\left[ P\left( y|y_{1}^{n}\right) -I(y\geq y_{n+1})\right] ^{2}dy.
\label{Eq:CRPSloss}
\end{equation}%
The CRPS is sensitive to distance,\ rewarding the assignment of high
predictive mass near to the realised value of $y_{n+1}$.

The fourth scoring rule is the left tail weighted CRPS (TWCRPS) proposed in 
\cite{gneiting2011comparing}, which is defined as 
\begin{equation}
S_{\text{TWCRPS}}\left[ P(Y_{n+1}|y_{1}^{n}),y_{n+1}\right] =-\int_{0}^{1}2%
\left[ I(P^{-1}(\alpha |y_{1}^{n})\geq y_{n+1})-\alpha \right] \left[
P^{-1}(\alpha |y_{1}^{n})-y_{n+1}\right] (1-\alpha )^{2}d\alpha .
\label{Eq:TWCRPSloss}
\end{equation}%
This score penalises more heavily longer distances to realised values that
are observed in the left tail.

The last score that we consider is the interval score (IS) proposed in \cite%
{gneiting2007strictly}. The IS formula is defined over the $100\left(
1-\alpha \right) $\% prediction interval, and given by 
\begin{flalign*}
S_{\text{IS}}\left[ P(Y_{n+1}|y_{1}^{n}),y_{n+1}\right] =&-\bigg{\{}
u_{n+1}-l_{n+1}+\frac{2}{\alpha }\left( l_{n+1}-y_{n+1}\right) \boldsymbol{1}%
\{y_{n+1}<l_{n+1}\}\\&+\frac{2}{\alpha }\left( y_{n+1}-u_{n+1}\right) 
\boldsymbol{1}\{y_{n+1}>u_{n+1}\}\bigg{\}} ,
\end{flalign*}%
where $l_{n+1}$ and $u_{n+1}$ denote the $100\left( \frac{\alpha }{2}\right) 
$\% and $100\left( 1-\frac{\alpha }{2}\right) $\% predictive quantile,
respectively. This score rewards high predictive accuracy of the $100\left(
1-\alpha \right) $\% predictive interval with $0<\alpha <1$. In this paper
we set $\alpha =0.05$.

\section{Technical results}
\subsection{Proofs of main results}

\begin{proof}[Proof of Lemma \ref{lem:first}]
	The proof follows a modification of the standard arguments; see, e.g., Theorem 3.2 in \cite{pakes1989simulation}. Fix $\epsilon>0$. By continuity of $\theta\mapsto H(\theta)$, there exists $\delta>0$ such that
	$$
	\text{Pr}\left[d(\widehat\theta_n,\theta_0)\ge\epsilon\right]\le \text{Pr}\left[H(\theta_0)-H(\widehat\theta_n)\ge\delta\right],
	$$	where $H(\theta)=\plim_n \ell_n(\theta)$, $\ell_n(\theta)=\frac{1}{n}\log p_\theta(\y)$ and $\theta_0$ satisfies $H(\theta_0)\ge\sup_{\theta\in\Theta}H(\theta)$. The stated result then follows if the RHS is $o(1)$. Since $\hat\kappa_n:=\Upsilon_n(\widehat\theta_n)/n\ge0$ for all $n$,
	\begin{flalign*}
	H(\theta_0)-H(\widehat\theta_n)&\leq H(\theta_0)-H(\widehat\theta_n)+\hat\kappa_n\le 2\sup_{\theta\in\Theta}|H(\theta)-\ell_n(\theta)|+\ell_n(\theta_0)-\ell_n(\widehat\theta_n)+\hat\kappa_n
	\end{flalign*}By Assumption \ref{ass:mle}, the first term is $o_p(1)$, and we can concentrate on the second term. From the definition of $\widehat\theta_n$, and since $0<p(\theta)<\infty$ for all $\theta$,
	\begin{flalign*}
	\left[\ell_n(\theta_0)-\kappa_n+\frac{1}{n}\log p(\theta_0)\right]=	\left[\ell_n(\theta_0)-\kappa_n\right]+o(1)\le\left[\ell_n(\widehat\theta_n)-\hat\kappa_n+\frac{1}{n}\log p(\widehat\theta_n)\right]=\left[\ell_n(\widehat\theta_n)-\hat\kappa_n\right]+o(1).
	\end{flalign*}Therefore,  
	\begin{flalign*}
	\ell_n(\theta_0)-\ell_n(\widehat\theta_n)+\hat\kappa_n&=	[\ell_n(\theta_0)-\kappa_n]-[\ell_n(\widehat\theta_n)-\hat\kappa_n]+\kappa_n\le o(1)+\kappa_n.
	\end{flalign*}Conclude that $H(\theta_0)-H(\widehat\theta_n)\leq o_p(1)$ if $\kappa_n=o_p(1)$.
\end{proof}

\begin{proof}[Proof of Lemma \ref{lem:second}]
	The proof follows along the same lines used to prove results for generalized posteriors. See, in particular, \cite{chernozhukov2003mcmc}, \cite{miller2019asymptotic}, and \cite{syring2020gibbs}. 
	
	Define $\Pi_n(\Theta):=\int_\Theta \exp\left\{\widehat{{L}}_n(\theta)\right\}p(\theta)\dt\theta$ and recall that, by hypothesis, for all $n\ge1$, $\Pi_n(\Theta)<\infty$. Fix $\epsilon>0$, and let $A_\epsilon:=\{\theta:d(\theta,\theta_\star)>\epsilon\}$. For any $\delta>0$, 
	\begin{flalign*}
	\widehat{Q}(A_\epsilon|\y)&=\frac{\Pi_n(A_\epsilon)}{\Pi_n(\Theta)}=\frac{\Pi_n(A_\epsilon)\exp\left\{-\widehat{L}(\theta_\star)+n\delta\right\}}{\Pi_n(\Theta)\exp\left\{-\widehat{L}_n(\theta_\star)+n\delta\right\}}\\&=\frac{\int_{A_\epsilon}\exp\left\{-\widehat{L}(\theta_\star)+n\delta\right\}\exp\left\{\widehat{L}_n(\theta)\right\}p(\theta)\dt\theta}{\int_\Theta \exp\left\{-\widehat{L}_n(\theta_\star)+n\delta\right\}\exp\left\{\widehat{L}_n(\theta)\right\}p(\theta)\dt\theta}\\&=\frac{N_n}{D_n}.
	\end{flalign*}	
	We treat the numerator and denominator separately. 
	
	Write the numerator as
	\begin{flalign*}
	N_n=\int_{A_\epsilon}\exp\left\{n\left[\widehat{L}_n(\theta)/n-\widehat{L}_n(\theta_\star)/n+\delta\right]\right\}p(\theta)\dt\theta. 
	\end{flalign*}	Considering $\widehat{L}_n(\theta)/n-\widehat{L}_n(\theta_\star)/n$, we have that 
	\begin{flalign*}
	\widehat{L}_n(\theta)/n-\widehat{L}_n(\theta_\star)/n&\leq 2\sup_{\theta\in\Theta,\lambda\in\Lambda}|\mathcal{L}_n(\theta,\lambda)/n-\mathcal{L}(\theta,\lambda)|+\mathcal{L}[\theta,\widehat{\lambda}_n(\theta)]-\mathcal{L}[\theta_\star,\widehat{\lambda}_n(\theta_\star)]\\&\le o_p(1) +\left\{\mathcal{L}[\theta,\widehat{\lambda}_n(\theta)]-\mathcal{L}[\theta,\lambda(\theta)]\right\}-\left\{\mathcal{L}[\theta_\star,\widehat{\lambda}_n(\theta_\star)]-\mathcal{L}[\theta_\star,\lambda(\theta_\star)]\right\}\\&+\mathcal{L}[\theta,\lambda(\theta)]-\mathcal{L}[\theta_\star,\lambda(\theta_\star)]\\&\le o_p(1)+\mathcal{L}[\theta,\lambda(\theta)]-\mathcal{L}[\theta_\star,\lambda(\theta_\star)]\\&\le o_p(1)-\delta
	\end{flalign*}
	where the first inequality follows from the triangle inequality, the second from Assumption \ref{ass:post}(2.b), and the third follows from consistency of $\widehat{\lambda}_n(\theta)$, uniformly over $\theta$, Assumption \ref{ass:post} (1), and the last follows from the identification condition in Assumption \ref{ass:post}(2.a). Thus for any $\epsilon>0$,  
	$$
	\liminf_{n\rightarrow\infty}\text{Pr}\left[\sup_{\theta:d(\theta,\theta_\star)>\epsilon}\frac{1}{n}\left\{\widehat{L}_n(\theta)-\widehat{L}_n(\theta_\star)\right\}\le-\delta\right]=1.
	$$Therefore, for any $\theta\in A_\epsilon$, 
	$$
	\frac{1}{n}\left\{\widehat{L}_n(\theta)-\widehat{L}_n(\theta_\star)\right\}+\delta\le0,
	$$with probability converging to one (wpc1), so that for all $n$ large enough 
	$$
	\exp\left\{n\left[\widehat{L}_n(\theta)/n-\widehat{L}_n(\theta_\star)/n+\delta\right]\right\}\le 1.
	$$
	Consequently, wpc1, 
	\begin{flalign*}
	N_n=\int_{A_\epsilon}\exp\left\{n\left[\widehat{L}_n(\theta)/n-\widehat{L}_n(\theta_\star)/n+\delta\right]\right\}p(\theta)\dt\theta\leq \int_{A_\epsilon}p(\theta)\dt\theta\le1. 
	\end{flalign*}
	
	To bound the denominator, first define $L(\theta)=\mathcal{L}[\theta,\lambda(\theta)]$ and $G_\delta:=\{\theta:L(\theta)-L(\theta_\star)<-\delta/2\}$. For any $\theta\in G_\delta$, by Assumption \ref{ass:post}(2.b), 
	$$
	\left\{\widehat{L}_n(\theta)/n-\widehat{L}_n(\theta_\star)/n\right\}+\delta/2\rightarrow L(\theta)-L(\theta_\star)+\delta/2<0,
	$$wpc1. Thus, for any $\delta>0$ and any $\theta\in G_\delta$, 
	$
	\exp\left\{n\left[\widehat{L}_n(\theta)/n-\widehat{L}_n(\theta_\star)/n+\delta\right]\right\}\rightarrow\infty
	$ as $n\rightarrow\infty$ wpc1. By Fatou's Lemma
	\begin{flalign*}
	\liminf_{n\rightarrow\infty}\exp\left\{-\widehat{L}_n(\theta_\star)+n\delta\right\}\Pi_n(G_\delta)&=\liminf_{n\rightarrow\infty}\int_{G_\delta}\exp\left\{n\left[\widehat{L}_n(\theta)/n-\widehat{L}_n(\theta_\star)/n+\delta\right]\right\}p(\theta)\dt\theta\\&\ge\liminf_{n\rightarrow\infty}\exp\left[n\delta/4\right]\int_{G_\delta}p(\theta)\dt\theta.
	\end{flalign*}Since $\int_{G_\delta}p(\theta)\dt\theta>0$ for any $\delta>0$, by Assumption \ref{ass:post}(3), the term on the RHS of the inequality diverges as $n\rightarrow\infty$. Use the fact that 
	$$
	\Pi_n(\Theta)\ge \Pi_n(G_\delta)
	$$to deduce that $D_n\rightarrow\infty$ as $n\rightarrow\infty$ (wpc1). 
\end{proof}

\begin{proof}[Lemma \ref{lem:lgm}]
	The complete data likelihood is proportional to 
	\begin{flalign*} p(x_1^n, y_1^n| \theta)=& \{2 \pi \sigma_{0}^{2}\}^{-n} \exp \left\{-\frac{1}{2 \sigma_{0}^{2}} \sum_{k=1}^{n-1}\left(x_{k+1}-\rho x_{k}\right)^{2}-\frac{1}{2 \sigma_{0}^{2}} \sum_{k=1}^{n}\left(y_{k}-\alpha x_{k}\right)^{2}-\frac{1}{2 \sigma_{0}^{2}}\left(x_{1}\right)^{2}\right\}\\&=\{2\pi\sigma^2_0\}^{-n}\exp\left\{-\frac{1}{2\sigma_0^2}\left[(\x)'\Omega_n(\theta)\x-2\alpha(\y)'\x+(\y)'\y\right]\right\},
	\end{flalign*}for the matrix 
	$$
	\Omega_n(\theta):=\begin{pmatrix}
	(1+\rho^2+\alpha^2)&-\rho&0&\dots&0\\-\rho&(1+\rho^2+\alpha^2)&-\rho&\dots&0\\\vdots&\vdots&\vdots&\vdots&\vdots\\0&0&0&-\rho&(1+\alpha^2)
	\end{pmatrix}.
	$$
	The states can be analytically integrated out, using known results for multivariate normal integrals, to obtain the observed data likelihood $p(\y|\theta)$:
	$$
	p(\y|\t)=\{2\pi\sigma_0^2\}^{-n}\left[\frac{(2\pi)^n}{|\sigma_0^{-2}\Omega_n(\theta)|}\right]^{1/2}\exp\left\{-\frac{1}{2\sigma_0^2}\left[(\y)'\y-\alpha^2(\y)'\Omega_n(\theta)^{-1}\y\right]\right\},
	$$which yields the observable data log-likelihood
	\begin{flalign*}
	\log p(\y|\t)=-\frac{n}{2}\log2\pi-\frac{n}{2}\log(\sigma_0^2) -\frac{1}{2}\log |\Omega(\theta)|+\frac{1}{2\sigma_{0}^2}\left[\alpha^2 (y_1^n)'\Omega_n(\theta)^{-1}y_1^n-(y_1^n)'y_1^n\right].
	\end{flalign*}

	Following Lemma \ref{lem:first}, consider the infeasible situation where our variational family for $\theta$ is
	$$
	\mathcal{Q}_\theta:=\{q_\theta:\delta_{\theta_0}(t),\;t\in\Theta\}.
	$$Under this choice, consistency follows if 	$
	\Upsilon_n(q)/n=\frac{1}{n}\left\{\log p(\y|\t_0)-\mathcal{L}_n(\t_0)\right\}=o_p(1).
	$ In the remainder, we drop the dependence of $q_x(\x|\t)$ on $\t_0$ and simply denote $q_x(\x)=q_x(\x|\t_0)$.

	Under the choice of $\mathcal{Q}_x$,  
	\begin{flalign*}
	\mathcal{L}_n(\theta_0)=\int_{\mathcal{X}}q_x(\x)\log \frac{p(\x,\y|\theta)}{q_x(\x)}\dt\x=&-n\log2\pi-n\log\sigma_0^2\\&-\frac{1}{2\sigma_0^2}\sum_{k=2}\int (x_k-\rho_0 x_{k-1})^2 q_{x}(x_k,x_{k-1})\dt x_k\dt x_{k-1}\\&- \frac{1}{2\sigma_0^2}\sum_{k=1}\int (y_k-\alpha_0 x_{k})^2 q_{x}(x_{k+1},x_{k})\dt x_{k+1}\dt x_{k}\\&-\frac{1}{2\sigma_0^2}\int (x_1)^2\mathcal{N}(x_1;0,\sigma_0^2)\dt x_1\\&-\int q_x(\x)\log\prod_{k=2}q(x_k|x_{k-1})\dt\x, 
	\end{flalign*}where each of the above individual pieces can be solved explicitly:
	\begin{flalign*}
	&\int q_x(x_k,x_{k-1})\log q(x_k|x_{k-1})\dt x_k \dt x_{k-1}=-\frac{1}{2}\log2\pi-\frac{1}{2}-\frac{1}{2}\log \sigma_0^2(1-\rho_0^2) \\
	&\int (x_k-\rho_0 x_{k-1})^2 q_{x}(x_k,x_{k-1})\dt x_k\dt x_{k-1}=\sigma_{0}^2(1-\rho_0^2)\\
	&\int (x_1)^2\mathcal{N}(x_1;0,\sigma_0^2)\dt x_1=\sigma_0^2\\\
	&\int (y_k-\alpha x_{k})^2q_x(x_k)\dt x_{k}=y_k^2+\alpha^2\sigma_{0}^2,
	\end{flalign*}to obtain
	\begin{flalign*}
	\mathcal{L}_n(\theta_0)=&-\frac{n}{2}\log2\pi-\frac{n}{2}\log\sigma_0^2+\frac{n}{2}+\frac{n}{2}\log(1-\rho_0^2)-\frac{1}{2\sigma_0^2}\{n\sigma^2_0(1-\rho^2_0)\}-\frac{1}{2\sigma_0^2}(\y)'\y-\frac{n}{2}\alpha_0^2
	\end{flalign*}
	Similarly, we have that 
	$$
	\log p(\y|\t_0)=-\frac{n}{2}\log2\pi-\frac{n}{2}\log\sigma_0^2-\frac{1}{2}\log|\Omega_n(\theta_0)|+\frac{\alpha_0^2}{2\sigma_0^2}(\y)'\Omega_n(\theta_0)^{-1}\y-\frac{1}{2\sigma_0^2}(\y)'\y
	$$
	and	
	Jensen's Gap is
	\begin{flalign*}
	\Upsilon_n(q)=&-\frac{1}{2}\log|\Omega(\theta_0)|+\frac{1}{2\sigma_0^2}[\alpha^2_0(\y)'\Omega(\theta_0)^{-1}\y+n\alpha_0^2\sigma^2_0]-\frac{n}{2}-\frac{n}{2}\log(1-\rho^2_0)+\frac{1}{2}\{n(1-\rho_0^2)\}
	\end{flalign*}
	
	To determine whether $\Upsilon_n(q)/n=o_p(1)$, we must first consider the behavior of the first and second terms in $\Upsilon_n(q)$. For the first term, we note that $|\Omega_n(\theta_0)|$ is a deterministic function of $\theta_0$ and, it can be shown that, {see Lemma \ref{lem:mat} for details,} if 
	$
	(1+\alpha^2_0+\rho^2_0)^2-4\rho^2_0\neq 0
	$,
	then, for $a= (1+\alpha^2_0+\rho^2_0)$ and $d:=\sqrt{a_0^2-4\rho_0^2}$,
	$$
	|\Omega_n(\theta_0)|=\frac{1}{{d}}\left(\left(\frac{a+{d}}{2}\right)^{n+1}-\left(\frac{a-{d}}{2}\right)^{n+1}\right), 
	$$and we can define 
	$$
	C_1(\theta_0):=\plim_{n\rightarrow\infty}\log|\Omega_n(\theta_0)|/n .
	$$Conversely, if $(1+\alpha_0^2+\rho_0^2)-4\alpha_0^2=0$, then 
	$$
	|\Omega_n(\theta_0)|=(n+1)(a/2)^{n},
	$$and we can define $C_1(\theta_0)$ similarly in this case. 
	
	Now, consider the second term in $\Upsilon_n(q)$. From the structure of the model, for $\sigma_x^2=\sigma_0^2/(1-\rho_0^2)$, 
	$$
	\y\sim \mathcal{N}\left(0,M \right),\;M:=\sigma^2_x[(\sigma^2_0/\sigma^2_x)I+V],\;V^{-1}:=\begin{pmatrix}
	(1+\rho_0^2)&-\rho_0&0&\dots&0\\-\rho_0&(1+\rho_0^2)&-\rho_0&\dots&0\\\vdots&\vdots&\vdots&\vdots&\vdots\\0&0&0&-\rho_0&(1+\rho^2_0)
	\end{pmatrix},
	$$so that we can conclude that, for any $n\ge2$,  
	$$
	\mathbb{E}[(\y)'\Omega_n^{-1}(\theta_0)\y]=\text{Tr}\left[\Omega_n(\theta_0)^{-1}M\right],\;\text{Var}\left[(\y)'\Omega_n(\theta_0)^{-1}\y\right]=2\text{Tr}\left[\Omega_n(\theta_0)^{-1}M\Omega_n(\theta_0)^{-1}M\right].
	$$
	
	Define $Z_n=(\y)'\Omega_n(\theta_0)^{-1}\y$ and apply Markov's inequality to $Z_n$ to obtain, for any $\epsilon>0$,   
	\begin{flalign}\label{eq:lgm1}
	\text{Pr}\left(|Z_n-\mathbb{E}[Z_n]|>n\epsilon\right)&\leq \text{Var}[Z_n]/(n^2\epsilon^2)=\text{Tr}\left[\Omega_n(\theta_0)^{-1}M\Omega_n(\theta_0)^{-1}M\right]/(n^2\epsilon^2)
	\end{flalign}In addition, for $D_1$ denoting the first diagonal element of $\Omega_n(\theta_0)^{-1}M\Omega_n(\theta_0)^{-1}M$, and $D_n$ the $n$-th, 
	\begin{flalign}\label{eq:lgm2}
	\text{Var}[Z_n]&=\text{Tr}\left[\Omega_n(\theta_0)^{-1}M\Omega_n(\theta_0)^{-1}M\right]\le n \cdot\sup_{n\ge1}\left\{|D_1|,\dots,|D_n|\right\}.
	\end{flalign}Define the sequence $c_n:=\sup_{n\ge1}\left\{|D_1|,\dots,|D_n|\right\}$. For any $0\le\rho_0^2<1$ and $0\le|\alpha_0|<M<\infty$, the sequence $c_n$ is non-random and bounded for each $n$, hence we have that $c_n/n\rightarrow0$. From the boundedness of $c_n$, apply equations \eqref{eq:lgm1} and \eqref{eq:lgm2} to conclude that, for any $\epsilon>0$, 
	$$\lim_{n\rightarrow\infty}\frac{\text{Var}[Z_n]}{(n\epsilon)^2}\le \lim_{n\rightarrow\infty} \frac{\sup_{n\ge1}| \text{diag}\left\{\Omega_n(\theta_0)^{-1}M\Omega_n(\theta_0)^{-1}M\right\}|}{n\epsilon^2}=\lim_{n\rightarrow\infty}\frac{c_n}{n\epsilon^2}\rightarrow0. $$The above argument and equation \eqref{eq:lgm1} allow us to conclude that $C_2(\theta_0):=\frac{\alpha_0^2}{2\sigma_0^2}\lim_{n\rightarrow\infty}\text{Tr}[\Omega_n(\theta_0)M]/n$ exists and that 
	$$
	\plim_{n\rightarrow\infty}\frac{\alpha_0^2}{2\sigma_0^2 n}(\y)'\Omega_n(\theta_0)^{-1}\y=C_2(\theta_0).
	$$	
	
	We are now ready to specialize the above to the two cases of interest. 
	
	\bigskip
	
	\noindent\textbf{Case 1:} If $\rho_0=0$, then $|\Omega_n(\theta_0)|=(1+\alpha_0^2)^{n}$, and $\log|\Omega_n(\theta_0)|=n\log(1+\alpha_0^2)$. In addition, $M=2\sigma^2_0I$ and $\Omega_n(\theta_0)^{-1}=\frac{1}{(1+\alpha_0^2)}I_n$, with $I_n$ the $n$-dimensional identity matrix, so that $\text{Tr}(\Omega_n(\theta_0)^{-1}M)=n\frac{2\sigma^2_0}{1+\alpha_0^2}$. Therefore, we have that $$C_1(\theta_0)=\frac{1}{2}\log(1+\alpha_0^2)\text{ and }C_2(\theta_0)=\alpha^2_0/(1+\alpha_0^2).$$ Since $n^{-1}\log p_{\theta_0}(\y)\rightarrow_p H(\theta_0)\ge0$, which minimizes entropy, and since $\Upsilon_n(q)\ge0$, consequently, VI for $\alpha_0$ will be consistent iff  
	$$
	\plim_{n\rightarrow\infty}\Upsilon_n(q)/n=0=-\frac{1}{2}\log(1+\alpha_0^2)+\frac{\alpha_0^2}{1+\alpha_0^2}+\alpha_0^2.
	$$Over $0\le|\alpha_0|<M$, $M$ finite, the above equation has the unique solution $\alpha_0=0$. 
	
	\bigskip 
	
	
	
	
	\medskip 
	\noindent\textbf{Case 2: $\alpha_0=0$.} Similar to the above, since $H(\theta_0)$ is entropy minimizing, and since $\Upsilon_n(q)/n\ge0$, it must be that $\Upsilon_n(q)/n=o_p(1)$ if VI is to be consistent. However, if $\alpha_0=0$, we have that 
	$$
	\Upsilon_n(q)=-\frac{1}{2}\log |\Omega_n(\theta_0)|-\frac{n}{2}-\frac{n}{2}\log(1-\rho_0^2)+\frac{n}{2}(1-\rho_0^2).
	$$Apply Lemma \ref{lem:mat} in the supplementary material to obtain $|\Omega_n(\theta_0)|$ with $a=(1+\rho_0^2)$, and $b=c=-\rho_0$. In particular, use the fact, for $0\le|\rho_0|<1$, $d=\sqrt{a^2-4bc}=1-\rho_0^2$, and note that $a+d=2$ and $a-d=2\rho^2_0$, which allows us to specialize the general result in Lemma \ref{lem:mat} as 
	\begin{flalign*}
	|\Omega_n(\theta_0)|&=\frac{1}{(1-\rho_0^2)}\left\{\left[1-\left(\rho^{2}_0\right)^n\right]-\rho^2_0\left[1-\left(\rho^2_0\right)^{n-1}\right]\right\}=\frac{1}{1-\rho^2_0}\left\{1-\rho^2 -\rho_0^{2n}+\rho^2(\rho^{2(n-1)})\right\}\\&=1
	\end{flalign*}
	Conclude that  VI consistent is iff
	$$
	\plim_{n\rightarrow\infty}\Upsilon_n(q)/n=0=-\frac{1}{2}\log(1+\rho_0^2)-\frac{1}{2}-\frac{1}{2}\log(1-\rho_0^2)+\frac{1}{2}(1-\rho_0^2).
	$$The only solution to the above equation is $\rho_0=0$.

\end{proof}

\begin{proof}[Proof of Lemma \ref{lem:three}]
	Recall the complete data likelihood from the proof of Lemma \ref{lem:lgm}:
	\begin{flalign*} p(x_1^n, y_1^n| \theta)=& \{2 \pi \sigma_{0}^{2}\}^{-n} \exp \left\{-\frac{1}{2 \sigma_{0}^{2}} \sum_{k=1}^{n-1}\left(x_{k+1}-\rho_0x_{k}\right)^{2}-\frac{1}{2 \sigma_{0}^{2}} \sum_{k=1}^{n}\left(y_{k}-\alpha x_{k}\right)^{2}-\frac{1}{2 \sigma_{0}^{2}}\left(x_{1}\right)^{2}\right\}\\&=\{2\pi\sigma^2_0\}^{-n}\exp\left\{-\frac{1}{2\sigma_0^2}\left[(\x)'\Omega_n(\theta)\x-2\alpha(\y)'\x+(\y)'\y\right]\right\}. 
	\end{flalign*}	
	In this case, we  calculate 
	\begin{flalign*}
	\mathcal{L}_n(\theta,\lambda)&=\int_{\mathcal{X}}q_\lambda(\x)\log\frac{p(x_1^n, y_1^n| \theta)}{q_\lambda(\x)}\dt\x\\&=\int_{\mathcal{X}}q_\lambda(\x)\log p(x_1^n, y_1^n| \theta)\dt\x-\int_{\mathcal{X}}q_\lambda(\x)\log q_\lambda(\x)\dt\x.
	\end{flalign*}Writing the two terms as $\mathcal{L}_{1,n}(\theta,\lambda)$ and $\mathcal{L}_{2,n}(\theta,\lambda)$, let us focus on the first term. This can be rewritten as
	\begin{flalign*}
	\mathcal{L}_{1,n}(\theta,\lambda)&=-n\log 2\pi-\frac{1}{2}\E_{\x}[(\x)'\Omega_n(\theta)\x]-\alpha(\y)'\E_{\x}[\x]-\frac{1}{2}(\y)'\y\\&=-n\log 2\pi-\frac{1}{2}(\y)'\y-\frac{1}{2}\text{Tr}[\Omega_n(\theta)\nu(\lambda)\Phi_n(\lambda)]\\&=-n\log 2\pi-\frac{1}{2}(\y)'\y-\frac{\nu(\lambda)}{2}\text{Tr}[\Omega_n(\theta)\Phi_n(\lambda)],
	\end{flalign*}where the second equation comes from the fact that $\E_{\x}[\x]=0$, under $q_{\lambda}(\x)$, and properties of quadratic forms, and the third follows from linearity of $text{Tr}(\cdot)$. Tedious algebraic calculations show that  $$\text{Tr}[\Omega_n(\theta)\Phi_n(\lambda)]=(n-1)(1+\alpha^2+\rho^2-\rho\lambda)+(1+\alpha^2-\rho\lambda)$$and so we obtain 
	\begin{flalign*}
	\mathcal{L}_{1,n}(\theta,\lambda)&=-n\log 2\pi-\frac{1}{2}(\y)'\y-\frac{(n-1)\nu(\lambda)}{2}(1+\alpha^2+\rho^2-\rho\lambda)-\frac{\nu(\lambda)}{2}(1+\alpha^2-\rho\lambda)
	\end{flalign*}
	
	The second term can be written as 
	\begin{flalign*}
	\mathcal{L}_{2,n}(\theta,\lambda)&=\int_{\mathcal{X}}q_\lambda(\x)\log q_\lambda(\x)\dt\x=\E_{\x}[\log q_\lambda(\x)]=-\frac{n}{2}\log 2\pi-\frac{1}{2}\log |\nu(\lambda)\Phi_n(\lambda)|-n/2\\&=-\frac{n}{2}\log 2\pi-\frac{n}{2}-\frac{1}{2}\log\frac{\sigma^2_0}{(1-\lambda^2_\rho)}\\&=-\frac{n}{2}\log 2\pi-\frac{n}{2}-\frac{n}{2}\log {\sigma^2_0}+\frac{1}{2}\log{(1-\lambda^2_\rho)}
	\end{flalign*}where we have used the fact that for the matrix $\Phi_n(\lambda)$, $|\Phi_n(\lambda)|=(1-\lambda^2)^{n-1}$, so that we have 
	$$
	|\nu(\lambda)\Phi_n(\lambda)|=\nu(\lambda)^n|\Phi_n(\lambda)|=\frac{\lambda_\sigma^n}{(1-\lambda^2_\rho)^n}(1-\lambda^2)^{n-1}.
	$$

	Dividing these terms by $n$, and taking $n\rightarrow\infty$ yields the following limit criterion
	\begin{flalign*}
	\mathcal{L}(\theta,\lambda) = -\log2\pi-\frac{1}{2}\E_{\theta_0}\left[(\y)'(\y)\right]-\frac{\nu(\lambda)}{2}(1+\alpha^2+\rho^2-\rho\lambda). 
	\end{flalign*}Differentiating $\mathcal{L}(\theta,\lambda)$ with respect to $\lambda$ and solving yields two solutions:
	\begin{flalign}
	\lambda(\theta)=\begin{cases}\frac{\alpha^2+\rho^2-\sqrt{\left(\alpha^2+\rho^2+\rho+1\right)\,\left(\alpha^2+\rho^2-\rho+1\right)}+1}{\rho}\\ \frac{\alpha^2+\rho^2+\sqrt{\left(\alpha^2+\rho^2+\rho+1\right)\,\left(\alpha^2+\rho^2-\rho+1\right)}+1}{\rho} \end{cases}
	\end{flalign}

	It can be shown that the first solution is the maximum, while the second is the minimum. Using the function $\lambda(\theta)$, the concentrated objective function is 
	\begin{flalign}
	\mathcal{L}[\theta,\lambda(\theta)]=\frac{\sqrt{\left(\alpha^2+\rho^2+\rho+1\right)\,\left(\alpha^2+\rho^2-\rho+1\right)}}{2\,\left(\frac{{\left(\alpha^2+\rho^2-\sqrt{\left(\alpha^2+\rho^2+\rho+1\right)\,\left(\alpha^2+\rho^2-\rho+1\right)}+1\right)}^2}{\rho^2}-1\right)}.
	\end{flalign}
	On the compact space $[\underline{\alpha},\overline{\alpha}]\times [\underline{\rho},\overline{\rho}]$, the function $ \mathcal{L}[\theta,\lambda(\theta)]$ is continuous and bounded. Hence, by the extreme value theorem $ \mathcal{L}[\theta,\lambda(\theta)]$ achieves its maximum at some point in $[\underline{\alpha},\overline{\alpha}]\times [\underline{\rho},\overline{\rho}]$.

	To prove that $\theta^\star=(0,\underline\rho)'$, we can consider the following two cases:
	\begin{enumerate}
		\item $\alpha^\star\ge 0$, and $\rho^\star= \underline{\rho}$;
		\item $\rho^\star\ge \underline{\rho}$, and $\alpha^\star= 0$;
	\end{enumerate}
	
	\noindent\textbf{Case 1:} Take $\rho^\star=\underline{\rho}>0$, but $\underline{\rho}$ close to zero. Under this choice, we can approximate $ \mathcal{L}[\theta,\lambda(\theta)]$ 
	as
	\begin{flalign}
	\mathcal{L}[\theta,\lambda(\theta)]|_{\rho=\rho^\star}\approx\frac{\sqrt{\left(\alpha^2+1\right)\,\left(\alpha^2+1\right)}}{2\,\left(\frac{{\left(1+\alpha^2-\sqrt{\left(\alpha^2+1\right)\,\left(\alpha^2+1\right)}\right)}^2}{\underline{\rho}^2}-1\right)}=-\frac{1}{2}{(1+\alpha^2)}{}.
	\end{flalign}The RHS of the above is a decreasing function of $\alpha$, so that its maximum over $[0,\overline\alpha]$ is attained at $\alpha=0$. 
	
	\bigskip
	
	\noindent\textbf{Case 2:} Taking $\alpha^\star=0$, we have 
	\begin{flalign}\label{eq:rhofun}
	\mathcal{L}[\theta,\lambda(\theta)]|_{\alpha=0}= \frac{\sqrt{\left(\rho^2+\rho+1\right)\,\left(\rho^2-\rho+1\right)}}{\frac{2\,{\left(\rho^2-\sqrt{\left(\rho^2+\rho+1\right)\,\left(\rho^2-\rho+1\right)}+1\right)}^2}{\rho^2}-2}=\frac{1}{2}\frac{\sqrt{\rho^4+\rho^2+1}}{\left[\frac{(1+\rho^2-\sqrt{\rho^4+\rho^2+1})^2}{\rho^2}-1\right]}.
	\end{flalign}For $0<\underline\rho\le\rho\le\overline\rho<1$, the denominator of the RHS is always larger than $-1$, and always less than $(2-\sqrt{3})^2-1\approx-.9282$. It can verified (e.g., numerically) that the above function is monotonically decreasing over $[\underline\rho,\overline\rho]$, so that its maximum  is attained at $\rho=\underline\rho$.
	
	\medskip
	
	Hence, the maximum of $\mathcal{L}[\theta,\lambda(\theta)]$ over $[0,\overline{\alpha}]\times [\underline{\rho},\overline{\rho}]$ is given by $\theta^\star=(0,\underline\rho)'$. 
	
\end{proof}

\begin{proof}[Proof of Corollary \ref{cor:one}]The result is a direct consequence of Lemma \ref{lem:three} and known results. For two multivariate normal distributions (with the same dimension $d\ge1$) $N(\mu_1,\Sigma_1)$ and $N(\mu_2,\Sigma_2)$, the KL divergence is
	\begin{equation}\label{eq:kldiv}
	\text{KL}\left[N(\mu_1,\Sigma_1)||N(\mu_2,\Sigma_2)\right]	=\frac{1}{2}\left[\log \frac{\left|\Sigma_{2}\right|}{\left|\Sigma_{1}\right|}-d+\operatorname{tr}\left\{\Sigma_{2}^{-1} \Sigma_{1}\right\}+\left(\mu_{2}-\mu_{1}\right)^{T} \Sigma_{2}^{-1}\left(\mu_{2}-\mu_{1}\right)\right]
	\end{equation} For known $\theta$, the posterior of the latent states can be obtained via the Kalman filter, and has the known form:
	$$
	\pi(x_n|\y,\theta_0)=N\{x_n; {x}_{n|n}(\theta_0),P_{n|n}(\theta_0)\},
	$$where ${x}_{n|n}$ and $P_{n|n}$ are known functions that are obtained from the Kalman filter recursions. Let $\lambda_\star=\lambda(\theta_\star)$, and note that since $q_{\lambda_\star}(\x)$ is multivariate Gaussian with mean $0$ and variance $\Phi_n(\lambda_\star)$, we immediately obtain that the marginal density of $x_n$ under the variational family is $q_{\lambda_\star}(x_n)=N\{x_n;0,\sigma^2_0/(1-\lambda_\star^2)\}$. Applying equation \eqref{eq:kldiv} then yields
	\begin{flalign*}
	\text{KL}\left[\pi(x_n|\y,\theta_0)||q_{\lambda_\star}(x_n)\right]	&=\text{KL}\left[N\{ {x}_{n|n}(\theta_0),P_{n|n}(\theta_0)\}||N\{0,\sigma^2_0/(1-\lambda_\star^2)\}\right]\\&=\frac{1}{2}\left[\log\frac{\sigma^2_0/(1-\lambda_\star^2)}{P_{n|n}(\theta_0)}-1+\frac{P_{n|n}(\theta_0)}{\sigma^2_0/(1-\lambda_\star^2)}+\frac{ {x}^2_{n|n}(\theta_0)}{\sigma^2_0/(1-\lambda_\star^2)}\right]\\&=\frac{1}{2}\left[-\log ([1-\lambda_\star^2]{P_{n|n}(\theta_0)})-1+[1-\lambda_\star^2]\left[{P_{n|n}(\theta_0)}+ {x}^2_{n|n}\right]\right]\\&=\frac{1}{2}\left[ \log \frac{\exp\{[1-\lambda_\star^2]P_{n|n}(\theta_0)\}}{[1-\lambda_\star^2]P_{n|n}(\theta_0)}-1+[1-\lambda_\star^2]x^2_{n|n}(\theta_0)\right]
	\end{flalign*}where the third inequality follows from the fact that $\sigma^2_0=1$, and the last from re-arranging terms. 
	
	For any $y\ge C>0$, differentiating the function $\exp(y)/y$ and solving for its zero yields 
	$$
	\frac{\exp(y)}{y}[1-({1}/{y})]=0\iff [1-({1}/{y})]=0
	$$which yields the unique solution $y=1$ on $y\ge C>0$. A second round of differentiation shows that this function is positive at $y=1$. Hence, $\exp(y)/y$ attains a unique minimum at $y=1$, and we have $\exp(y)/y\ge\exp(1)$ for all $y\ge C>0$. Consequently, $\log \{\exp(y)/y\}\ge1$ and we have shown that the first term in the KL divergence is positive when $P_{n|n}(\theta_0)>0$. Since $[1-\lambda_\star^2]x_{n|n}^2(\theta_0)\ge0$, the stated result follows.

\end{proof}
\subsection{Additional lemmas}

\begin{lemma}
	\label{lem:mat} Let 
	\begin{equation*}
	\Omega_n:= 
	\begin{pmatrix}
	a & c & 0 & \dots & 0 \\ 
	b & a & c & \dots & 0 \\ 
	\vdots & \vdots & \vdots & \vdots & \vdots \\ 
	0 & 0 & 0 & b & 1%
	\end{pmatrix}
	,\;a>0,\;a^2-4bc\ne0.
	\end{equation*}
	Then, for $d=\sqrt{a^2-4bc}$, 
	\begin{equation*}
	|\Omega_n|=\frac{1}{d}\left[\left(\frac{a+d}{2}\right)^{n}-\left(\frac{a-d}{
		2 }\right)^{n}\right]-bc\frac{1}{d}\left[\left(\frac{a+d}{2}
	\right)^{n-1}-\left(\frac{a-d}{2}\right)^{n-1}\right].
	\end{equation*}
\end{lemma}

\begin{proof}
	The determinant of tridiagonal matrices satisfy the following recurrence relationship: for $f_k=|\Omega_{k}|$, with $\Omega_{k}$ denoting the $k\times k$ matrix, $1< k\le n$,  
	$$f_n=a_nf_{n-1}-c_{n-1}b_{n-1}f_{n-2},$$ where $f_0=0$ and $f_1=1$, and $c_{k},b_{k}$ refer to the elements above and below, respectively, the diagonal term $a_n$. In this case, this relationship implies that $f_n=|\Omega_n|$ satisfies
	$$
	f_n=f_{n-1}-cb f_{n-2}.
	$$

	However, note that, for an $1\le k<n$, $f_{k}$ is actually a $k\times k$ dimensional Toeplitz matrix. Applying a Laplace expansion to $f_k$ twice yields the linear homogenous recurrence equation
	$$
	f_{k}=a f_{k-1} -bc f_{k-2},
	$$ 	which has characteristic polynomial $p(x)=x^2-ax+bc$ that admits two solutions
	$$
	x=\frac{a\pm\sqrt{a^2-4bc}}{2}
	$$ Under the condition that ${a^2-4bc}\ne0$, the roots are distinct and we have that 
	$$
	f_k = c_1 \left(\frac{a+\sqrt{a^2-4bc}}{2}\right)^{k}+c_2\left(\frac{a-\sqrt{a^2-4bc}}{2}\right)^k,$$
	for some $c_1$ and $c_2$ that satisfy the initial conditions of the recurrent relation. In particular, we have that $f_1=a$, and $f_2=a^2-bc$, so that 
	$$
	a^2-bc=a(f_1)-bc(f_0)=a^2-bc(f_0),
	$$which implies that $d_0=1$. Consequently, $c_1+c_2=1$. Letting $d=\sqrt{a^2-4bc}$, we see that the case of $k=1$ implies
	\begin{flalign*}
	2a=k_1 (a+d)+k_2(a-d)=a+(c_1-c_2)d&\implies c_1=c_2+a/d
	\\&\implies 1=2c_2+a/d\\&\implies c_2=\frac{d-a}{2d}=-\frac{1}{d}\frac{(a-d)}{2}\\&\implies c_1 =\frac{d-a+2a}{2d}=\frac{1}{d}\frac{a+d}{2}
	\end{flalign*}
	
	Therefore, we can conclude that
	$$
	f_k=\frac{1}{d}\left[\left(\frac{a+d}{2}\right)^{k+1}-\left(\frac{a-d}{2}\right)^{k+1}\right],
	$$and we then have closed form expressions for the determinants $f_{n-1}$ and $f_{n-2}$. Plugging in these definitions
	\begin{flalign*}
	f_n &= f_{n-1}-bcf_{n-2}\\&=\frac{1}{d}\left[\left(\frac{a+d}{2}\right)^{n}-\left(\frac{a-d}{2}\right)^{n}\right]-bc\frac{1}{d}\left[\left(\frac{a+d}{2}\right)^{n-1}-\left(\frac{a-d}{2}\right)^{n-1}\right].
	\end{flalign*}
\end{proof}

\end{document}